\documentclass[oneside,english]{article}
\usepackage{mathtools}
\usepackage[T1]{fontenc}
\usepackage[latin9]{inputenc}
\usepackage{array}
\usepackage{amsthm}
\usepackage{amsmath}
\usepackage{amstext}
\usepackage{amsfonts}
\usepackage{amssymb}
\usepackage{graphicx}
\usepackage{setspace}
\usepackage{esint}
\usepackage{fullpage}
\usepackage{subcaption}
\usepackage{dsfont}
\usepackage{placeins}
\usepackage{fourier}

\usepackage[colorinlistoftodos]{todonotes}

\makeatletter

\theoremstyle{plain}
\newtheorem{thm}{\protect\theoremname}
\numberwithin{thm}{section}

\theoremstyle{plain}
\newtheorem{lemma}[thm]{\protect\lemmaname}

\theoremstyle{plain}
\newtheorem{corr}[thm]{\protect\corrname}

\theoremstyle{plain}
\newtheorem{obs}{\protect\obsname}
\numberwithin{obs}{section}

\theoremstyle{plain}
\newtheorem{remark}{\protect\remarkname}
\numberwithin{remark}{subsection}

\theoremstyle{definition}
\newtheorem{defn}{\protect\definitionname}
\numberwithin{defn}{section}

\theoremstyle{plain}
\newtheorem{claim}[thm]{\protect\claimname}

\AtBeginDocument{%
  \expandafter\renewcommand\expandafter\subsubsection\expandafter
    {\expandafter\@fb@secFB\subsubsection}%
  \newcommand\@fb@secFB{\FloatBarrier
    \gdef\@fb@afterHHook{\@fb@topbarrier \gdef\@fb@afterHHook{}}}%
  \g@addto@macro\@afterheading{\@fb@afterHHook}%
  \gdef\@fb@afterHHook{}%
}

\makeatother

\usepackage{babel}

\providecommand{\definitionname}{Definition}

\providecommand{\obsname}{Observation}
\providecommand{\theoremname}{Theorem}
\providecommand{\lemmaname}{Lemma}
\providecommand{\corrname}{Corollary}
\providecommand{\remarkname}{Remark}

\providecommand{\claimname}{Claim}

\title{Characterizing minimum-length coordinated motions for two discs}

\author{David Kirkpatrick and Paul Liu}

\date{}

\begin{document}
\maketitle

\begin{abstract}
We study the problem of determining optimal coordinated motions for two disc robots in an otherwise obstacle-free plane. Using the total path length traced by the two disc centres as a measure of distance, we give an exact characterization of a shortest collision-avoiding motion for all initial and final configurations of the robots. 
The individual paths are composed of at most six (straight or circular-arc) segments, and their total length can be expressed as a simple integral with a closed form solution depending only on the initial and final configuration of the robots.
Furthermore, the paths can be parametrized in such a way that (i) only one robot is moving at any given time (decoupled motion), or (ii) the angle between the two robots' centres changes monotonically.
\end{abstract}

\section{Introduction}

In this paper we consider the problem of planning collision-free motions for two disc robots of arbitrary radius in an otherwise obstacle-free environment. Given two discs $\mathbb{A}$ and $\mathbb{B}$ in the plane, with specified initial and final configurations, we seek a shortest collision-free motion taking $\mathbb{A}$ and $\mathbb{B}$ from their initial to their final configurations. The length of such a motion is defined to be the length sum of paths traced by the centres of $\mathbb{A}$ and $\mathbb{B}$. 

The consideration of disc robots in motion planning has amassed a substantial body of research, the bulk of which is focused on the feasibility, rather than optimality, of motions. 
Schwartz and Sharir \cite{sharir0} were the first to study motion planning for $k$ discs among polygonal obstacles with $n$ total edges. For $k=2$, they developed an $\mathcal{O}(n^3)$ algorithm (later improved to $\mathcal{O}(n^2)$ \cite{sharir2, yap}) to determine if a collision-free motion connecting two specified configurations is feasible. When the number of robots $k$ is unbounded, Spirakis and Yap \cite{spirakis} showed that determining feasibility is strongly NP-hard for disc robots, although the proof relies on the robots having different radii. For the analogous problem with rectangular robots, determining feasibility is PSPACE-hard, as shown by Hopcroft et al. \cite{hopcroft0} and Hopcroft and Wilfong \cite{hopcroft1}. This result was later generalized by Hearn and Demaine \cite{hearn} for rectangular robots of size $1\times 2$ and $2 \times 1$.

On the practical side, heuristic and sampling based algorithms have been employed to solve motion planning problem for up to hundreds of robots \cite{gildardo,standley,wagner}. These algorithms typically use standard search strategies such as $A*$ coupled with domain specific heuristics (see \cite{planning-text} and the references contained therein). While efficient in practice, these algorithms are typically numerical or iterative in nature, with no precise performance bounds. A variety of alternative cost measures for our problem has also been considered, such as the minimum time motion under velocity constraints \cite{ChenIerardi, ladder, turpin} as well as the motion minimizing the total number of continuous movements \cite{abellanas, bereg, dumitrescu}. 

A variant of our problem is when the robots are homogeneous and unlabeled. In this case, any robot is allowed to move to any target location, so long as each target position is covered by exactly one robot. For $k=2$ discs, the unlabeled case is trivial as one can apply our labeled algorithm twice. However, when $k$ is unbounded, Solovey and Halperin \cite{solovey2} show that the unlabeled problem is PSPACE-hard, even in the case of unit squares with polygonal obstacles. Surprisingly, when the robots are located within a simple polygon with no obstacles, a polynomial time for checking feasibility exists \cite{adler0}. As in the labeled case, a variety of cost measures has been explored for the unlabeled case. Solovey et al. \cite{solovey1} gives an $\tilde{\mathcal{O}}(k^4+k^2n^2)$ algorithm that minimizes the length sum of paths traced by the centres of the discs with additive error $4k$. In work by Turpin et al. \cite{turpin}, an optimal solution is found in polynomial time when the cost function is the maximum path length traversed by any single robot. However, their algorithm requires that the working space is obstacle free and the initial locations of the robots are far enough apart. 

This paper makes several novel contributions to the understanding of minimum-length coordinated motions. For the case of two arbitrary discs, we first characterize all initial and final configurations that admit straight-line optimal motions. A special case of this, of course, is where the final configuration is a simple translate of the initial configuration. For all other initial and final configurations, the motion from initial to final configuration involves either a net clockwise or counter-clockwise turn in the relative position of the discs. In this case, our results describe either (i) a single optimal motion, or (ii) two feasible motions, of which one is optimal among all net clockwise motions and the other is optimal among all net counter-clockwise motions. The motions that we describe have bounded curvature except at a constant number of isolated points; in fact, they are composed of a constant number of straight segments and circular arcs, of radius $s$, the sum of the disc radii. The path length itself can be expressed as a simple integral depending only on the initial and final positions of the discs. Moreover, all paths that we describe can can be realized by two different kinds of coordinated motion: \emph{coupled} or \emph{decoupled}. In the coupled motion, the angle formed by a ray joining the two disc centres changes monotonically. Furthermore, the two discs are in contact for a connected interval of time. That is, once the two discs move out of contact, they are never in contact again. In the decoupled motion, only one of the discs moves at any given time. 

Our general approach is based on the Cauchy surface area formula, which was first applied to motion planning by Icking et al.~\cite{icking} to establish the optimality of motions of a directed line segment in the plane, where distance is measured by the length sum of the paths traced by the two endpoints of the segment. This problem has a rich history, and was first posed by Ulam \cite{ulam} and subsequently solved by Gurevich \cite{gurevich}. Other approaches to that of Icking et al. are quite different, and use control theory to obtain differential equations that characterize the optimal motion \cite{gurevich, verriest}. Of course, the problem of moving a directed line segment of length $s$ corresponds exactly to the coordinated motion of two discs with radius sum $s$ constrained to remain in contact throughout the motion. Hence the coordinated motion of two discs with radius sum $s$ can also be seen as the problem of moving an ``extensible" line segment that can extend freely but has minimum length $s$. As such, our results also generalize those of Icking et al.  Although we use some of the same tools introduced by Icking et al., our generalization is non-trivial; the doubling argument that lies at the heart of the proof of Icking et al.~depends in an essential way on the assumption that the rod length is fixed throughout the motion.

The rest of the paper is organized as follows. In Section \ref{sec:backgrd} we outline some basic definitions as well as our tools for the problem. In Section \ref{sec:overview} we summarize the general structure of our proofs, with the main proof and algorithm given in Sections \ref{sec:mainprf} and \ref{sec:mainprf2}. 

%%%%%%%%%%%%%%%%%%%%%%%%%%%%%%%%%%%%%%%%%%%%%%%%%%
\section{Background}
\label{sec:backgrd}
 
To describe the motion of a pair of disc robots between their initial and final configurations, we first make precise several terms that have intuitive meaning. We assume for concreteness that the radii of the two discs sum to $s$.

%The position of the discs will be described with \emph{placements}. Placements are essentially a pair of locations (points in $\Re^2$) denoting the instantaneous positions of the centres of discs ${\mathbb A}$ and ${\mathbb B}$. 

\begin{defn}
The (instantaneous) \textbf{position} of a disc is simply a point in $\Re^2$ specifying the location of its centre.
A \textbf{placement} of a disc pair $({\mathbb A},{\mathbb B})$ is a pair $(A, B)$, where $A$ (resp. $B$) denotes the position of ${\mathbb A}$ (resp. ${\mathbb B}$). A placement $(A,B)$ is said to be \textbf{compatible} if $||A-B||\geq s$. 
\end{defn}

A pair of discs can move from placement to placement through a motion, which we can now define:

\begin{defn}
A \textbf{trajectory} $\xi_{\mathbb A}$ of a disc ${\mathbb A}$ from a position $A_0$ to a position $A_1$ is a continuous, rectifiable curve of the form $\xi_{\mathbb A} : [0,1] \rightarrow \Re^2$, where $\xi_{\mathbb A}(0) = A_0$, $\xi_{\mathbb A}(1) = A_1$.

A \textbf{(coordinated) motion} $m$ of a disc pair $({\mathbb A},{\mathbb B})$ from a placement $(A_0,B_0)$ to a placement $(A_1, B_1)$ is a pair $(\xi_{\mathbb A}, \xi_{\mathbb B})$, where $\xi_{\mathbb A}$ (resp. $\xi_{\mathbb B}$) is a trajectory of $\mathbb A$ (resp. $\mathbb B$) from position $A_0$ to $A_1$ (resp. position $B_0$ to $B_1$). A motion is said to be compatible or feasible if all of its associated placements are compatible.
\end{defn}

Since we are interested in characterizing collision-free motions, we will assume that, unless otherwise specified, all placements and motions that arise in this paper are compatible.

\begin{defn}
%The \textbf{separation} between two placements $(A_0,B_0)$ and $(A_1,B_1)$ is simply $||A_1-A_0||+||B_1-B_0||$, that is, the net Euclidean distance traversed by discs $\mathbb A$ and $\mathbb B$ in moving from $A_0$ and $B_0$ to $A_1$ and $B_1$.  
The \textbf{length} $\ell(\xi_{\mathbb A})$ of a trajectory $\xi_{\mathbb A}$ is simply the Euclidean arc-length of its trace, that is,
\[
	\ell(\xi_{\mathbb A}) = \sup_{T} \sum_{i=1}^k ||\xi_{\mathbb A}(t_{i-1})- \xi_{\mathbb A}(t_i)||
\]
where the supremum is taken over all subdivisions $T=\{t_0,t_1,\ldots,t_k\}$ of $[0,1]$ where $0 = t_0 < t_1 < \cdots < t_k = 1$. 

The \textbf{length} $\ell(m)$ of a motion $m$ is the sum of the lengths of its associated trajectories, i.e. $\ell(m) = \ell(\xi_{\mathbb A}) + \ell(\xi_{\mathbb B})$. Finally, the \textbf{(collision-free) distance} $d(P_0,P_1)$ between two placements $P_0=(A_0,B_0)$ and $P_1=(A_1, B_1)$ is the minimum possible length over all compatible motions $m$ from $P_0$ to $P_1$. We refer to any compatible motion $m$ between $P_0$ and $P_1$ satisfying $\ell(m) = d(P_0,P_1)$ as a \textbf{shortest} or \textbf{optimal} motion between $P_0$ and $P_1$. As a shorthand, we also use $\ell(C)$ to represent the perimeter of a closed curve $C$.
\end{defn}

The fact that $d$ is a metric on the set of placements is easy to check. Nevertheless, one may be concerned about the existence of  a shortest motion under this notion of distance.  The fact that a shortest motion exists is a consequence of the Hopf-Rinow theorem, for which details can be found in \cite{gromov}.

%%%%%%%%%%%%%%%%%%%%%%%%%%%%%%%%%%%%%%%%
\section{The general approach}
\label{sec:overview}

Suppose that the disc pair $(\mathbb{A}, \mathbb{B})$ has initial placement $P_0 = (A_0, B_0)$ and final placement $P_1 = (A_1, B_1)$, and let $m = (\xi_{\mathbb A}, \xi_{\mathbb B})$ be any motion from $P_0$ to $P_1$. Denote by $\widearc{\xi_{\mathbb A}}$ (resp. $\widearc{\xi_{\mathbb B}}$) the closed curve defining the boundary of the convex hull of $\xi_{\mathbb A}$ (resp. $\xi_{\mathbb B}$). Since
$\xi_{\mathbb A}$ (resp. $\xi_{\mathbb B}$), together with the segment $\overline{A_0A_1}$ (resp. $\overline{B_0B_1}$), forms a closed curve whose convex hull has boundary  $\widearc{\xi_{\mathbb A}}$ (resp. $\widearc{\xi_{\mathbb B}}$), it follows from convexity that:
\begin{equation}
\label{eq:inequal}
\ell(\xi_{\mathbb A})\geq \ell(\widearc{\xi_{\mathbb A}}) - |\overline{A_0A_1}|
\makebox[1cm]{and} 
\ell(\xi_{\mathbb B})\geq \ell(\widearc{\xi_{\mathbb B}}) - |\overline{B_0B_1}|.
\end{equation}
When the inequality for $\xi_{\mathbb A}$ (resp. $\xi_{\mathbb B}$) is an equality, we say that the trace of $\xi_{\mathbb A}$ (resp. $\xi_{\mathbb B}$) is convex. When both $\xi_{\mathbb A}$ and $\xi_{\mathbb B}$ are convex, we say that motion $m = (\xi_{\mathbb A}, \xi_{\mathbb B})$ is convex.

Given a placement $P=(A,B)$, we refer to the angle formed by the vector from $B$ to $A$ with respect to the $x$-axis as the \emph{angle} of the placement $P$. Let $[\theta_0, \theta_1]$ be the range of angles counter-clockwise between the angle of $P_0$ and $P_1$.

\begin{obs}
\label{monotone} Let $m$ be any motion from $P_0$ to $P_1$, and let $I$ be the range of angles realized by the set of placements in $m$.
Then $[\theta_0,\theta_1]\subseteq I$ or $S^1-[\theta_0,\theta_1]\subseteq I$, where $S^1 = [0,2\pi]$.
\end{obs}

We use Observation \ref{monotone} to categorize the motions we describe into \emph{net clockwise} and \emph{net counter-clockwise} motions. Net clockwise motions satisfy
$S^1-[\theta_0,\theta_1]\subseteq I$ and net counter-clockwise
motions satisfy $[\theta_0,\theta_1]\subseteq I$. 

Since any motion is either net clockwise or net counter-clockwise (or both) it suffices to optimize over net clockwise and net counter-clockwise motions separately.  The following lemma sets out sufficient conditions for net (counter-)clockwise motions to be optimal.

\begin{lemma}
\label{optkey}
Let $m = (\xi_{\mathbb A}, \xi_{\mathbb B})$ be any  net (counter-)clockwise motion from $P_0$ to $P_1$ satisfying the following properties:

\begin{enumerate}
	\item (Convexity) $\,\widearc{\xi_{\mathbb A}} = \xi_{\mathbb A} \cup \overline{A_0A_1}$ \makebox[1.3cm]{and}
$\widearc{\xi_{\mathbb B}} = \xi_{\mathbb B} \cup \overline{B_0B_1}$; \makebox[1.3cm]{and}
    \item (Minimality) $\,\ell(\widearc{\xi_{\mathbb A}}) + \ell(\widearc{\xi_{\mathbb B}})$ is minimized over all possible net (counter-)clockwise motions.
\end{enumerate}
Then $m$ is a shortest net (counter-)clockwise motion from $P_0$ to $P_1$.
\end{lemma}

\begin{proof}

Let $m' = ({\xi}'_{\mathbb A}, {\xi}'_{\mathbb B})$ be any net (counter-)clockwise motion from $P_0$ to $P_1$.
It follows from property 1 that 
$\ell(m) = \ell(\widearc{\xi_{\mathbb A}}) - |\overline{A_0A_1}| + \ell(\widearc{\xi_{\mathbb B}}) - |\overline{B_0B_1}|$.  
Furthermore, from 2 we know that 
$\ell(\widearc{\xi_{\mathbb A}}) + \ell(\widearc{\xi_{\mathbb B}}) \le 
\ell(\widearc{{\xi}'_{\mathbb A}}) + \ell(\widearc{{\xi}'_{\mathbb B}})$.
Thus, using inequality (\ref{eq:inequal}), we have $\ell(m) \leq  \ell(\widearc{{\xi}'_{\mathbb A}}) - |\overline{A_0A_1}| + \ell(\widearc{{\xi}'_{\mathbb B}}) - |\overline{B_0B_1}| \leq \ell(m').$
\end{proof}

\begin{figure}[ht]
\centering
\includegraphics[width=0.9\textwidth]{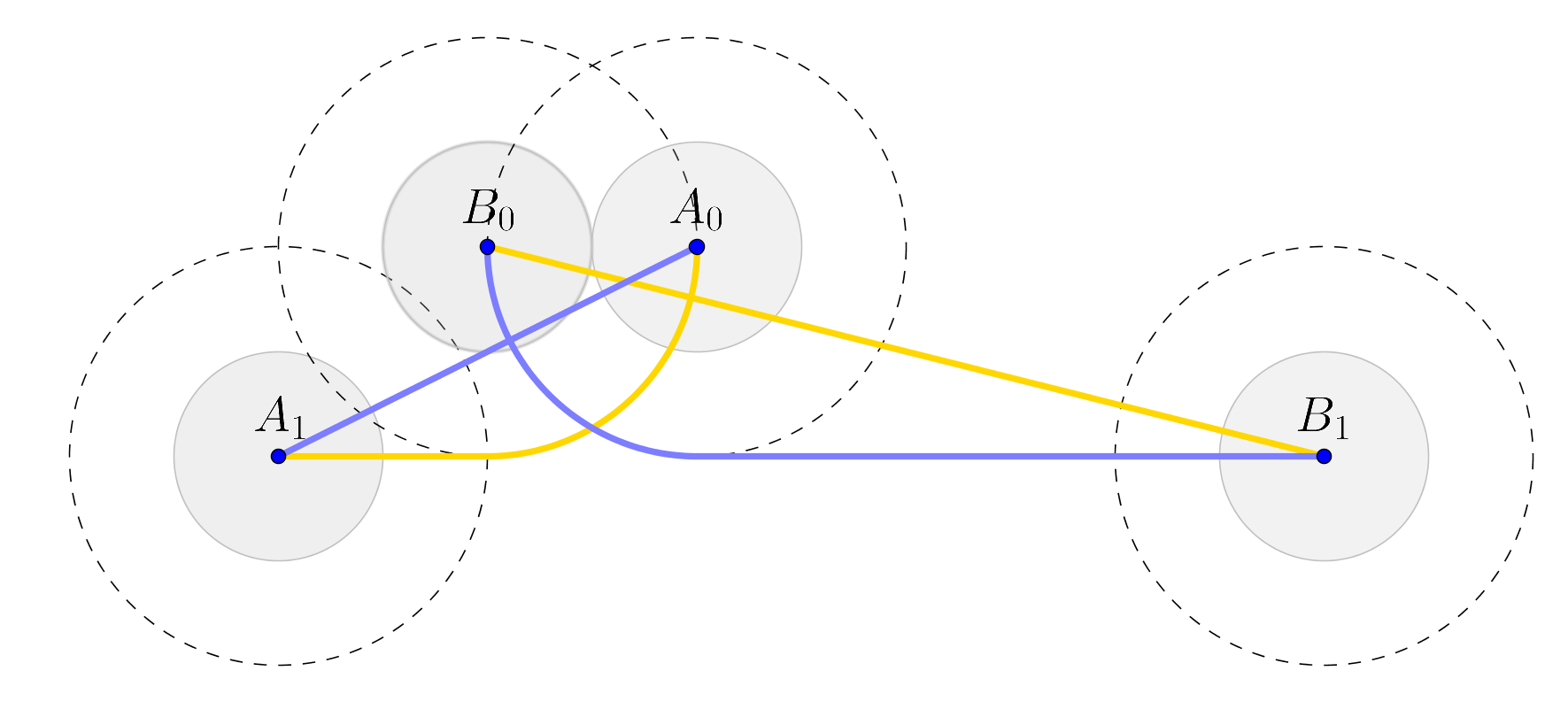}
\caption{Clockwise (yellow) and counter-clockwise (blue) motions satisfying the two properties of Lemma \ref{optkey}. \label{fig:twosol}}
\end{figure}

When a net (counter-)clockwise motion satisfies the two properties of Lemma \ref{optkey}, we say it is \emph{(counter-)clockwise optimal}. Figure \ref{fig:twosol} illustrates two motions from the placement $(A_0, B_0)$ to the placement $(A_1, B_1)$. The blue motion, where $\mathbb B$ first pivots about $A_0$ and moves to $B_1$, followed by $\mathbb A$ moving from $A_0$ to $A_1$, is counter-clockwise optimal. The yellow motion, where $ \mathbb A$ first pivots about $B_0$ and moves to $A_1$, followed by $\mathbb B$ moving from $B_0$ to $B_1$, is clockwise optimal (as one can check following the proofs of Sections \ref{sec:mainprf} and \ref{sec:mainprf2}). However, only the yellow motion is globally optimal.

While property 1 of Lemma \ref{optkey} is typically  easy to verify, property 2 is less straightforward and relies indirectly on an application of Cauchy's surface area formula (Theorem \ref{thm:cauchyorig}) as well as lower bounds we derive below. Theorem \ref{thm:cauchyorig} allows us to translate the problem of measuring lengths of curves into a problem of measuring the support functions of $\widearc{\xi_{\mathbb A}}$ and $\widearc{\xi_{\mathbb B}}$ at certain critical angles. Our approach is to lower bound these support functions to get a lower bound on the optimal path length, and then find a motion matching the lower bound.

\begin{defn}
\label{def:support}
Let $C$ be a closed curve. The \emph{support function} $h_C : S^1 \rightarrow \mathbb{R}$ of $C$ is defined as 
\[
h_C(\theta) = \sup \{ x \cos \theta + y \sin \theta : (x,y) \in C \}.
\]
%where $S^1$ denotes the range of angles of the unit circle.
For an angle $\theta$, the set points that realize the supremum above are called \emph{support points}, and the line oriented at angle $\frac{\pi}{2}+\theta$ going through the support points is called the \emph{support line} (see Figure \ref{fig:supports}).
\end{defn}

\begin{figure}[ht]
\centering
\includegraphics[width=0.4\textwidth]{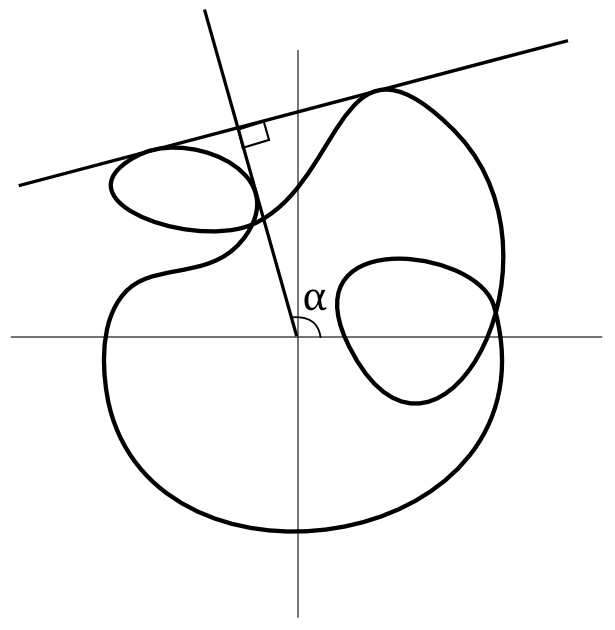}
\caption{The (two) support points and support line at angle $\alpha$ of a given curve. \label{fig:supports}}
\end{figure}

\begin{thm}
\label{thm:cauchyorig}
(Cauchy's surface area formula \cite[Section 5.3]{egg}) Let $C$ be a closed convex curve in the plane and $h_C$ be the support function of $C$. Then 
\begin{equation}
\ell(C) = \int_{0}^{2\pi} h_C(\theta)d\theta.
\end{equation}
\end{thm}

As noted in \cite{icking}, it follows from Theorem \ref{thm:cauchyorig} that we can bound the length of two convex curves in the plane:%, as is required in verifying property 2 of Lemma~\ref{optkey} .
\begin{corr}
\label{cor:cauchyorig}
Let $C_{1}$ and $C_{2}$ be closed convex curves in the plane. Then the
sum of their lengths can be expressed as follows: \label{thm:cauchy} 
\begin{equation}
\ell(C_{1})+\ell(C_{2}) = \int_{0}^{2\pi}\left(h_{1}(\theta)+h_{2}(\pi+\theta)\right)d\theta, \label{eq:cauchy}
\end{equation}
where $h_{i}$ is the support
function of $C_{i}$.
%The inequality is an equality when $C_{1}$ and $C_{2}$ are convex.
\end{corr}

In order to assert the optimality of our motions, we use the following observations that provide a bound on the support function of an arbitrary motion. Let $h_{\mathbb A}$ (resp. $h_{\mathbb B}$) denote the support function of $\wideparen{\xi_{\mathbb A}}$ (resp. $\wideparen{\xi_{\mathbb B}}$), and let $h_{\mathbb{AB}}(\theta)$ denote the sum $h_{\mathbb A}(\theta) + h_{\mathbb B}(\pi+\theta)$. Recall that $s$ is the radii sum of the two discs.

\begin{obs}
\label{lower-bound} Let $P_0$ and $P_1$ be two configurations and let $[\theta_0, \theta_1]$ be the range of angles counter-clockwise between the angles of $P_0$ and $P_1$. Then,
for all net counter-clockwise
motions from $P_0$ to $P_1$, and $\theta\in[\theta_0,\theta_1]$, 
$h_{\mathbb{AB}}(\theta)\geq s$.
Similarly, for all net clockwise motions and $\theta\in S^1-[\theta_0,\theta_1]$,
$h_{\mathbb{AB}}(\theta)\geq s$.
\end{obs}

\begin{obs}
\label{obs:point-wise}
For all support angles, the support function $h_{\mathbb A}$ (resp. $h_{\mathbb B}$) 
is lower bounded by the support function  $H_{\mathbb A}$ (resp. $H_{\mathbb B}$)
of $\overline{A_0A_1}$ (resp. $\overline{B_0B_1}$), since  $\overline{A_0A_1} \subset \wideparen{\xi_{\mathbb A}}$ 
(resp. $\overline{B_0B_1} \subset \wideparen{\xi_{\mathbb B}}$). From this, together with Observation~\ref{lower-bound}, it follows that the support function $h_{\mathbb{AB}}$ is lower bounded point-wise in the counter-clockwise and clockwise cases by
\begin{align}
& \max(H_{\mathbb{A}}(\theta)+H_{\mathbb{B}}(\pi+\theta), s\cdot\mathds{1}_{[\theta_0, \theta_1]}) \tag{net counter-clockwise} \\
& \max(H_{\mathbb{A}}(\theta)+H_{\mathbb{B}}(\pi+\theta), s\cdot\mathds{1}_{S^1-[\theta_0, \theta_1]}) \tag{net clockwise}
\end{align}
\end{obs}

In the next section we give explicit constructions of optimal motions for many initial-final configuration pairs. This includes, of course, all those whose associated trajectories correspond to two straight segments, what we refer to as
\textbf{straight-line motions}. In other cases, we construct both the clockwise and counter-clockwise optimal motions, one of which must be optimal among all motions.

%%%%%%%%%%%%%%%%%%%%%%%%%%%%%%%%%%

\section{Optimal paths for two discs}
\label{sec:mainprf}

Our constructions of shortest (counter-)clockwise motions can be summarized by the following theorem:
\begin{thm}
\label{thm:mainthm}
Let $\mathbb A$ and $\mathbb B$ be two discs with radius sum $s$ in an obstacle-free plane with arbitrary initial and final placements $P_0=(A_0, B_0)$ and $P_1= (A_1, B_1)$. Then there is a 
shortest motion from $P_0$ to $P_1$ whose associated trajectories are composed of at most six (straight or circular arcs of radius $s$) segments. 
\end{thm}

We devote this entire section to the identification  and exhaustive treatment of various cases of Theorem \ref{thm:mainthm}. The paths that we identify in each case also allow us to provide the following unified characterization of the optimal path length, covering all cases:
\begin{corr}
\label{cor:integral}
Let $H_\mathbb{A}$ and $H_\mathbb{B}$ be the support functions of the segments $\overline{A_0A_1}$ and $\overline{B_0B_1}$ respectively, $H_{\mathbb{AB}}(\theta):=H_{\mathbb{A}}(\theta)+H_{\mathbb{B}}(\pi+\theta)$, and $m$ be an optimal motion between $P_0$ and $P_1$. Let $[\theta_0,\theta_1]$ be the range of angles counter-clockwise between $P_0$ and $P_1$. Then 
\begin{equation*}
\ell(m) = \min\left( \int_0^{2\pi} \max(H_{\mathbb{AB}}(\theta), s\cdot\mathds{1}_{[\theta_0, \theta_1]}) \textrm{d} \theta, \int_0^{2\pi} \max(H_{\mathbb{AB}}(\theta), s\cdot\mathds{1}_{S^1-[\theta_0, \theta_1]}) \textrm{d} \theta \right) - |\overline{A_0A_1}| - |\overline{B_0B_1}|
\end{equation*}
where $\mathds{1}_{[a, b]}$ is the indicator function of the interval $[a,b]$.
\end{corr}

Though the expression in Corollary \ref{cor:integral} looks daunting, the only difference between the two integrals is the indicator function used. The support functions themselves can be expressed in closed form and the integrals are clearly lower bounds on the path length by Corollary \ref{thm:cauchy} and Observation~\ref{obs:point-wise}. We emphasize that the integrals can be expressed in closed form if needed, albeit with some cases involved.

We now introduce some additional tools that will help us classify the initial and final placements into different cases.

\begin{defn}
Let $p$ and $q$ be arbitrary points in the plane. 
\begin{enumerate}
\item[(a)] We denote by $s\text{-circ}(p)$ the circle of radius $s$ centred at point $p$.
\item[(b)] We denote by $s\text{-corr}(p,q)$ the  \textbf{$s$-corridor} associated with $p$ and $q$, defined to be the Minkowski sum of the line segment $\overline{pq}$ and an open disc of radius $s$.
\item[(c)] We denote by $s\text{-cone}(p,q)$ the cone formed by all half-lines from $p$ that intersect  $s\text{-circ}(q)$.
\end{enumerate}
\end{defn}

%\begin{figure}
%\includegraphics[width=\textwidth]{figures-ggb/cone-corr-defn}
%\caption{$s\text{-corr}(p,q)$ and $s\text{-cone}(r,p)$. \label{fig:cone-corr}
%\itodo{We should add shading to Figure\ref{fig:cone-corr}, or perhaps get rid of cones altogether if no longer needed.}}
%\end{figure}

If disc $\mathbb A$ is centred at location $A$ then $s\text{-circ}(A)$ corresponds to the locations forbidden to the centre of disc $\mathbb B$ in a compatible placement (see dotted circles in Figure~\ref{fig:twosol}). The corridors $s\text{-corr}(A_0,A_1)$ and $s\text{-corr}(B_0,B_1)$ play a critical role in partitioning initial and final placement pairs for which  straight-line trajectories (which are clearly optimal) are possible. Specifically, if point $A \not\in s\text{-corr}(B_0,B_1)$  then the line segment $\overline{B_0 B_1}$ does not intersect $s\text{-circ}(A)$; i.e. it is possible to translate $\mathbb B$ from $B_0$ to $B_1$ without interference from disc $\mathbb A$ with centre at point $A$. Similarly, if point $B \not\in s\text{-corr}(A_0,A_1)$ it is possible to translate $\mathbb A$ from $A_0$ to $A_1$ without interference from disc $\mathbb B$ with centre at point $B$. 

What follows is a case analysis of various scenarios for the
initial and final placements. We first classify the cases by the containment of $A_0$, $A_1$, $B_0$, $B_1$ within
$s\text{-corr}(A_0,A_1)$ and $s\text{-corr}(B_0,B_1)$ (cf. Table \ref{straight-line}). While there might appear to be
16 cases --- since each point is either contained within a corridor or not --- they cluster into just three disjoint collections, referred to as Cases 1, 2 and 3. These are further reduced by symmetries which 
include (i) interchanging the initial and final placements and (ii) switching the roles of $\mathbb A$ and $\mathbb B$. 

\begin{table}[ht]
\begin{center}
\begin{tabular}{|c|>{\centering}p{2cm}|>{\centering}p{2cm}|>{\centering}p{2cm}|>{\centering}p{2cm}|c|}
\hline 
 Case & $A_0\in s\text{-corr}(B_0,B_1)$ & $A_1\in s\text{-corr}(B_0,B_1)$ & $B_0\in s\text{-corr}(A_0,A_1)$ & $B_1\in  s\text{-corr}(A_0,A_1)$ & Type of motion\tabularnewline
\hline 
 1a & {\tt false} & {\tt *} & {\tt *} & {\tt false} & straight-line $(B_0\rightarrow B_1,A_0\rightarrow A_1)$\tabularnewline
\cline{2-5} 
 1b & {\tt *} & {\tt false} & {\tt false} & {\tt *} & %straight-line $(A_0\rightarrow A_1,B_0\rightarrow B_1)$
\tabularnewline
%Case 1 & {\tt false} & {*} & {*} & {\tt false} & straight-line $(B_0\rightarrow B_1,A_0\rightarrow A_1)$\tabularnewline
\hline 
 2a & {\tt true} & {\tt *} & {\tt true} & {\tt *} & See Section \ref{sec:newcase2} \tabularnewline
\cline{2-5} 
 2b    & {\tt *} & {\tt true} & {\tt *} & {\tt true} &  \tabularnewline
%\cline{2-5} 
\hline
 3a     & {\tt true} & {\tt true} & {\tt false} & {\tt false} & See Section \ref{sec:the-rest-2b} \tabularnewline
\cline{2-5} 
 3b    & {\tt false} & {\tt false} & {\tt true} & {\tt true} &  \tabularnewline
\hline 
\end{tabular}
\par\end{center}

\caption{All cases of possible motions. The  $*$ entries mean that the specified condition is unconstrained, i.e. it can be either $\tt true$ or $\tt false$. We leave some cases out due to symmetry. \label{straight-line}}
\end{table}

In all cases our specified motion has a common form -- with a possible interchange of the roles of $\mathbb A$ and $\mathbb B$. We identify an intermediate position $A_\text{int}$  (possibly $A_0$ or $A_1$) and perform the following sequence of (possibly degenerate) moves:

\begin{enumerate}
\item Move $\mathbb A$ on the shortest path from $A_0$ to $A_\text{int}$, avoiding $s\text{-circ}(B_0)$;
\item Move $\mathbb B$ on the shortest path from $B_0$ to $B_1$, avoiding $s\text{-circ}(A_\text{int})$; then
\item Move $\mathbb A$ on the shortest path from $A_\text{int}$ to $A_1$, while avoiding $s\text{-circ}(B_1)$.
\end{enumerate}

Without loss of generality, assume that our initial and final configurations have been \textbf{normalized} as follows: $B_0$ and $B_1$ lie on the $x$-axis with $B_0$ at the origin, $B_1$ right of $B_0$. In all but a few special cases, we will only examine motions that are net counter-clockwise; net clockwise optimal motions can be obtained by reflecting the initial and final placements across the $x$-axis and then examining net counter-clockwise motions. 

The net counter-clockwise orientation of our proposed motion $m=(\xi_{\mathbb A},\xi_{\mathbb B})$ as well as the convexity of $\xi_{\mathbb A}$ and $\xi_{\mathbb B}$ will typically be straightforward to verify.
To show the optimality of our motions, we show that the support function of our motions achieves the point-wise lower bound established in Observation~\ref{obs:point-wise}.
%Section~\ref{sec:overview}.

%[[ Should we refer explicitly to Observation 3.3?]]
%%%%%%%%%%%%%%%%%%%%%%%%%%
\subsection{Examples of counter-clockwise optimal motions}
\label{sec:special}
Before we attempt to identify optimal motions, it will be instructive to examine a special case, illustrated in Figure \ref{fig:special}. This case will provide the simplest non-trivial example of an optimal motion as well as an illustration to the form of our proofs.

\begin{figure}[h]
\centering
\includegraphics[width=0.8\textwidth]{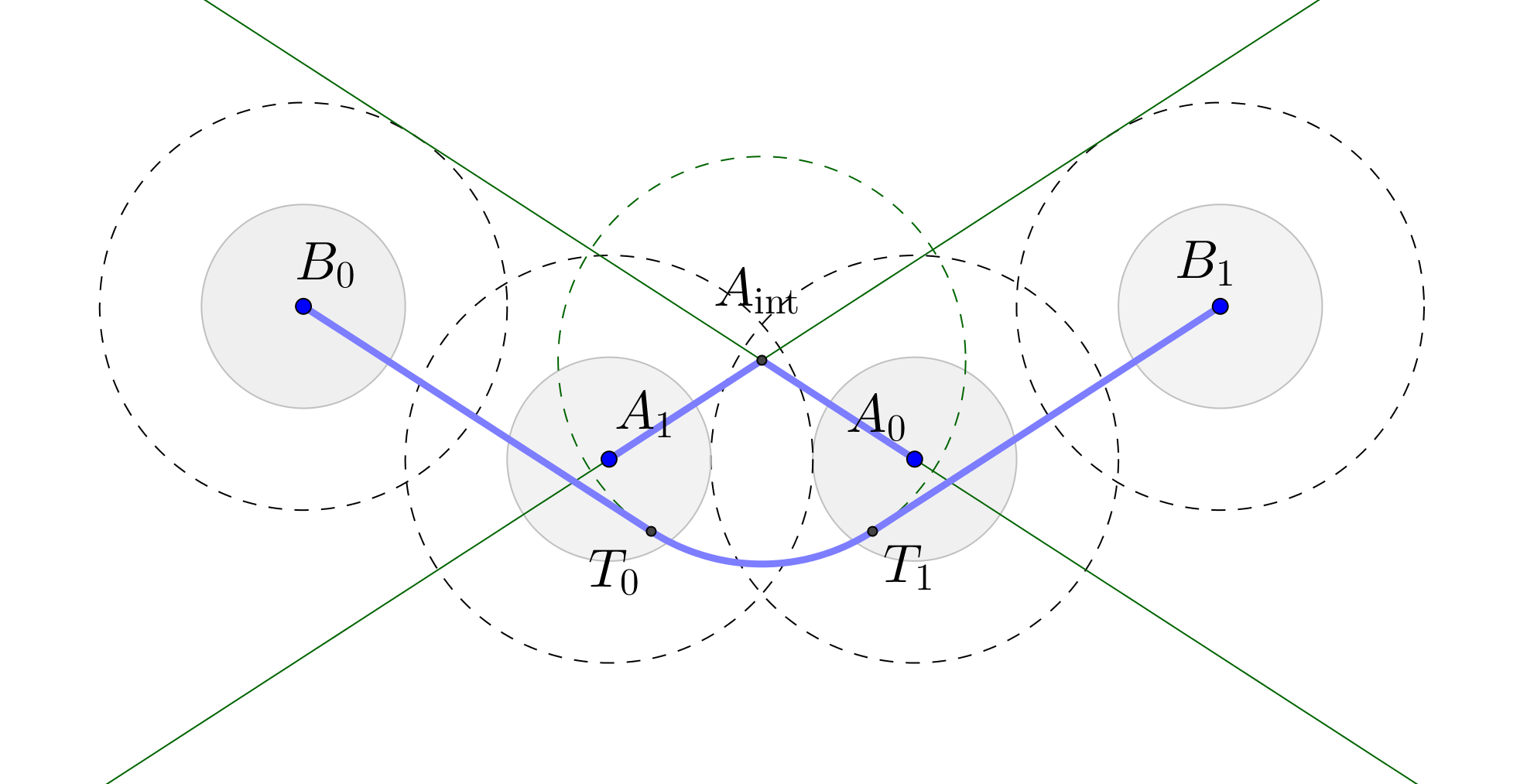}
\caption{\label{fig:special}}
\end{figure}

Consider the case shown in Figure \ref{fig:special}, where the $A_i$'s are in $s\text{-corr}(B_0,B_1)$, the $B_i's$ are on the $x$-axis, and the $A_i's$ are symmetric about the perpendicular bisector of $\overline{B_0B_1}$. 

%[[it is not so clear that the A's belong to the B-corridor in Figure 3]]

We define some points useful to our construction of the optimal motion. Let $a_0$ and $a_1$ be the upper tangents from $A_0$ to $s\text{-circ}(B_0)$, and $A_1$ to $s\text{-circ}(B_1)$ respectively. These two tangents intersect in a point $A_\text{int}$ 
%half-way between the two $B_i$'s. 
on the perpendicular bisector of $\overline{B_0B_1}$.
Let $b_0$ and $b_1$ be the lower tangents from $B_0$ and $B_1$ to $s\text{-circ}(A_\text{int})$ respectively, and let $T_0$ and $T_1$ be intersection points of $b_0$ and $b_1$ with $s\text{-circ}(A_\text{int})$. Note that by construction, $a_0$ is parallel to $b_1$ and $a_1$ is parallel to $b_0$.

\begin{claim}
\label{claim:special}
The following is a counter-clockwise optimal motion (see bolded outline in Figure~\ref{fig:special}):

\begin{enumerate}
\item Move $\mathbb A$ from $A_0$ to $A_\text{int}$;
\item Move $\mathbb B$ from $B_0$ to $B_1$, avoiding $s\text{-circ}(A_\text{int})$. This involves translating $\mathbb B$ from $B_0$ to $T_0$, rotating around $s\text{-circ}(A_\text{int})$ in a range of angles $[\beta_0, \beta_1]$, and finally translating from $T_1$ to $B_1$; then
\item Move $\mathbb A$ from $A_\text{int}$ to $A_1$.
\end{enumerate}

\end{claim}

\begin{proof}
It is easy to check that property 1 (convexity) of Lemma \ref{optkey} is satisfied. To show that property 2 (minimality) holds as well we verify that $h_{\mathbb{AB}}(\theta)$ matches its lower bound. By Observation~\ref{obs:point-wise}, we may check that for all angles $\theta$, either $h_{\mathbb{AB}}(\theta)=s$ for $\theta$ in the range of angles counter-clockwise between the initial and final placement, or is determined by $\mathbb A$ and $\mathbb B$ in their initial or final position. 

By construction, $\beta_{0}$ is normal to the orientation of $b_0$ (as well as $a_0$) and $\beta_{1}$ is normal to $b_1$ (as well as $a_1$). This ensures that for the range of angles $[\beta_{0},\beta_{1}]$,
$A_{\text{int}}$ is the support point of $h_\mathbb{A}(\theta)$ while the support point of $h_\mathbb{B}(\theta+\pi)$ lies on the arc of the circle traversed by $\mathbb B$. Hence $h_{\mathbb{AB}}(\theta)=s$ for $\theta \in[\beta_{0},\beta_{1}]$. 

Furthermore, $A_{\text{int}}$ is only a support point for angles in $[\beta_{0},\beta_{1}]$, since $\mathbb A$ moves along tangents $a_0$ and $a_1$. Thus for angles in $S^1-[\beta_{0},\beta_{1}]$, either $A_0$ or $A_1$ must be one support point, and either $B_0$ or $B_1$ must be the other.
\end{proof}

\begin{remark}
Even if the positions of $A_0$ and $A_1$ were swapped in the motion above, the trace of the optimal counter-clockwise motion would remain the same. The proof of optimality would proceed as above, using instead the tangents from $A_0$ to $B_1$ and $A_1$ to $B_0$ as $a_1$ and $a_0$ respectively.
\end{remark}

In the proof above, $A_\text{int}$ remains a support during the angles of $\mathbb B$'s rotation even if we shift it slightly vertically upwards.\footnote{However, the shifted $A_\text{int}$ is a support outside of the angles of rotation as well, which means $h_{\mathbb{AB}}$ does not achieve its lower bound outside of $[\beta_{0},\beta_{1}]$.} This motivates the following definition:

\begin{defn}
\label{def:dominating}
Let $p$ be a point in $s\text{-corr}(B_0, B_1)$. Let $\mathcal{R}$ be the region below both upper tangents from $p$ to $B_0$ and $B_1$. We call $\mathcal{R}$ the \textbf{dominated region} of $p$ with respect to $s\text{-corr}(B_0, B_1)$. For any point $q \in \mathcal{R}$, we say that $p$ \textbf{dominates} $q$.
\end{defn}

Note that if $p$ dominates $A_0$ and $A_1$, then substituting $p$ for $A_\text{int}$ in the proposed motion for Figure~\ref{fig:special} would maintain the property that the support function $h_{\mathbb{AB}}(\theta)$ is exactly $s$ in the angles of $\mathbb B$'s rotation. In fact, we have the following general lemma:

\begin{lemma}
\label{lem:support}
Let $p$ be any point that dominates $A_0$ and $A_1$, and let $m$ be any motion of the form:

\begin{enumerate}
\item Move $\mathbb A$ from $A_0$ to $p$ in a motion $m_1$, staying entirely within the region dominated by $p$;
\item Move $\mathbb B$ on the shortest path from $B_0$ to $B_1$ that travels below $s\text{-circ}(p)$. This involves moving $\mathbb B$ on a tangent segment $b_0$ from $B_0$ to $s\text{-circ}(p)$, rotating around $s\text{-circ}(A_\text{int})$ in a range of angles $[\beta_0, \beta_1]$, and moving on a tangent segment $b_1$ from $s\text{-circ}(p)$ to $B_1$; then
\item Move $\mathbb A$ from $p$ to $A_1$ in a motion $m_2$, staying entirely within the region dominated by $p$.
\end{enumerate}

For any such motion $m$, $h_{\mathbb{AB}}(\theta)=s$ for $\theta \in [\beta_0, \beta_1]$. Furthermore, if $m_1$ and $m_2$ form a convex trace when concatenated together, and the tangents of $m_1$ and $m_2$ at $p$ are parallel to $b_1$ and $b_0$ respectively, then $p$ is a support point iff the support angle is in the range $[\beta_0, \beta_1]$. 

\end{lemma}

\begin{proof}
The proof follows exactly the same analysis as the argument for $A_\text{int}$ in the proof of Claim~\ref{claim:special}, substituting the tangents of $m_1$ and $m_2$ at $p$ for $a_0$ and $a_1$.
\end{proof}

Lemma~\ref{lem:support} will allow us to exploit the commonality in many of the proofs we use in subsequent cases, as most motions will involve rotating around at least 1 pivot. 

As an example, consider Figure~\ref{fig:3rotation}, which shows an optimal counter-clockwise motion that we'll encounter in Case 2. In this motion, $\mathbb A$ first moves from $A_0$ to $A_\text{int}$, followed by $\mathbb B$ rotating from $B_0$ to $B_1$, and finished by moving $A_\text{int}$ to $A_1$. Lemma~\ref{lem:support} allows us to immediately say that $A_\text{int}$ is a support point exactly when $\mathbb B$ rotates from $B_0$ to $B_1$, since the motions from $A_0$ and $A_1$ to $A_{\text{int}}$ stay within the region dominated by $A_{\text{int}}$. The motion is also optimal, as the combined movement of $\mathbb A$ is convex, and the tangents at $A_\text{int}$ are parallel to the tangents of $s\text{-circ}(A_\text{int})$ at $B_0$ and $B_1$.

\begin{figure}[ht]
\centering
\includegraphics[width=0.7\textwidth]{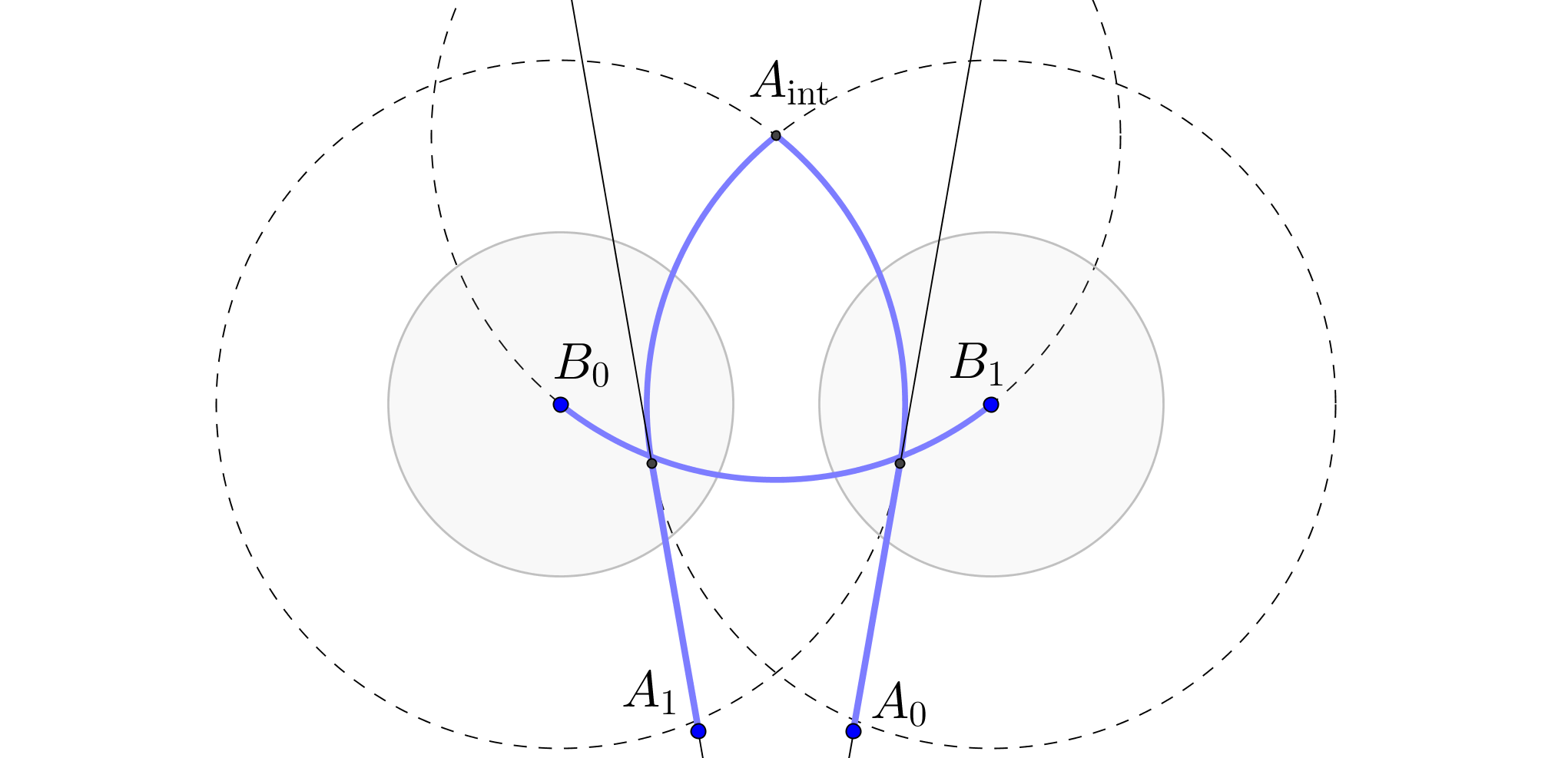}
\caption{\label{fig:3rotation}}
\end{figure}

%%%%%%%%%%%%%%%%%%%%%%%%%

\subsection{Certifying non-optimality of counter-clockwise motions}

The proofs we use in our case analysis will largely resemble the special case discussed in Section~\ref{sec:special}. Nevertheless for certain configurations, the tools we've developed in the previous section seem unable to show the optimality of net counter-clockwise motions. In such situations, we will show that the optimal net clockwise motion is shorter than any net counter-clockwise motion. 
In this section we analyse another special case, which will lead us to a set of placements for which we can prove that the optimal motion is net clockwise. This will help us deal with subcases for which the demonstration of net counter-clockwise optimal motions seems to be beyond the reach of our techniques.

\begin{figure}[ht]
\centering
\includegraphics[width=\textwidth]{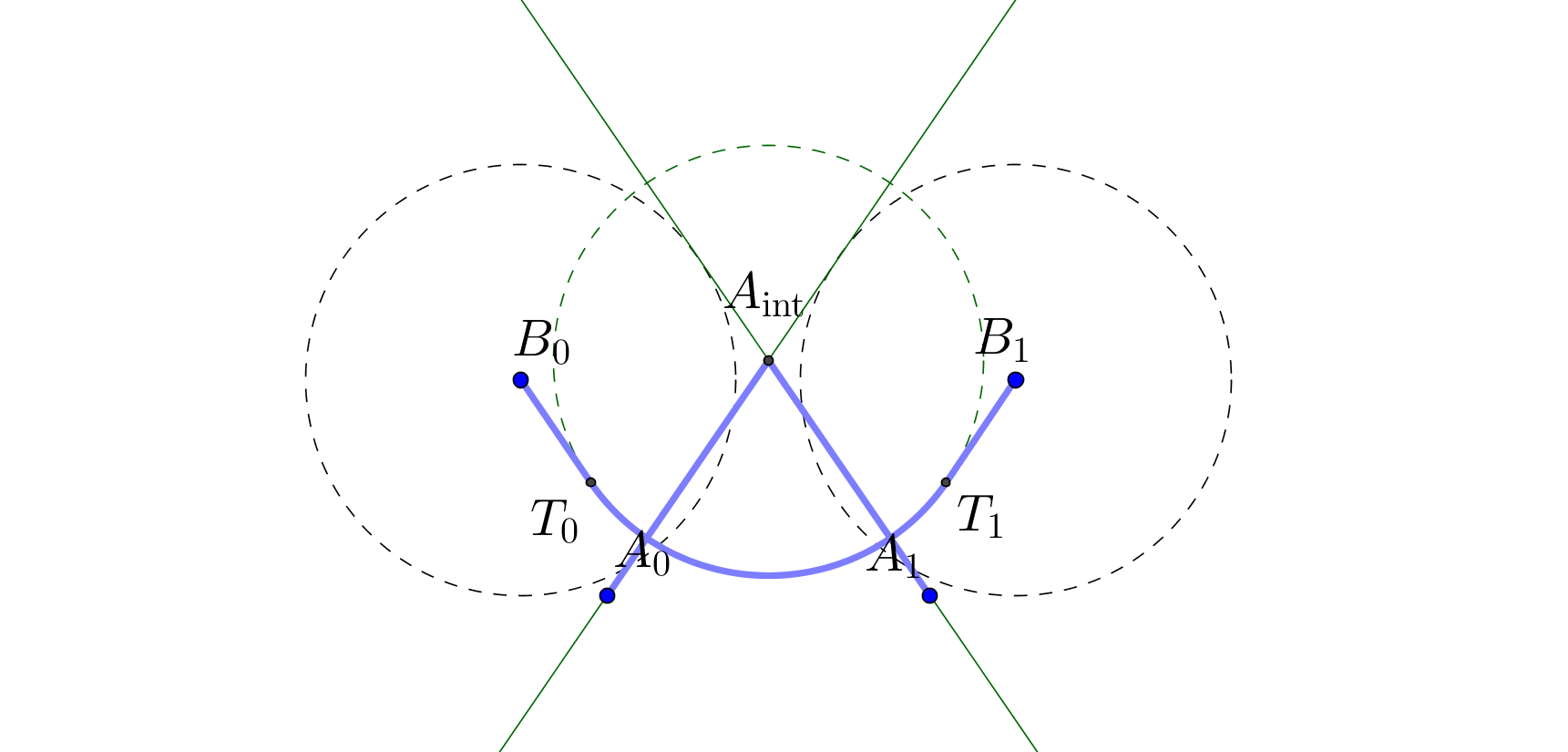}
\caption{\label{fig:blocked2}}
\end{figure}

Let us consider a variant of Figure~\ref{fig:special} where the $B_i$'s are now closer together, as depicted in Figure~\ref{fig:blocked2} and the positions of $A_0$ and $A_1$ are swapped. Again, we may draw the appropriate upper-tangents from the $A_i$'s and compute an intermediate point $A_\text{int}$. By Lemma~\ref{lem:support}, the trace length of the ``motion''  $m^\prime$ outlined in Figure~\ref{fig:blocked2} is no greater than that of any net-counterclockwise motion. However, the trace given in Figure~\ref{fig:blocked2} is not feasible, as it requires $A_0$ to move through $s\text{-circ}(B_0)$. In this case, we do not know of any counter-clockwise optimal motion for which optimality can be shown with Cauchy's surface area formula. As it turns out, we may sidestep this apparent difficulty by considering clockwise optimal motions. 

\begin{claim}
The optimal motion to Figure~\ref{fig:blocked2} is net clockwise. \label{claim:special-clockwise}
\end{claim}

\begin{proof}
Consider the trace shown in Figure~\ref{fig:blocked3}, where $A_{\text{int}}^\prime$ is the point $A_\text{int}$ reflected vertically across the segment $\overline{B_0B_1}$. The following motion $m$ is a feasible realization of this trace:

\begin{enumerate}
\item Move $\mathbb A$ from $A_0$ to the point $A^\prime$ vertically below $A_\text{int}^\prime$, on the along the segment $\overline{A_0A_1}$;
\item Move $\mathbb B$ from $B_0$ to $B_1$, rotating across the top of $s\text{-circ}(A_\text{int}^\prime)$; then
\item Move $\mathbb A$ to $A_1$.
\end{enumerate}

\begin{figure}[h]
\centering
	\includegraphics[width=0.7\textwidth]{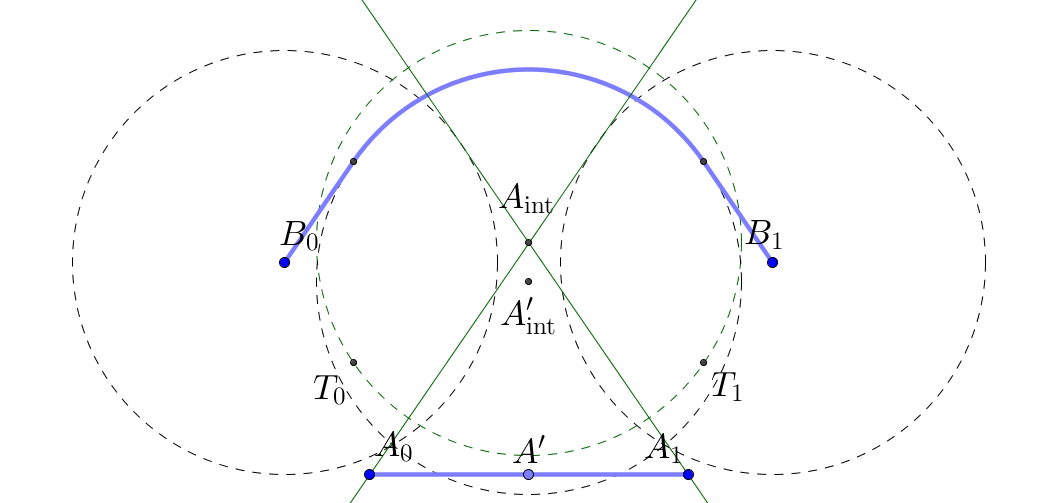}
\caption{\label{fig:blocked3}}
\end{figure}

It is easy to see that $\ell(m) < \ell(m^\prime)$: the total distance traveled by $\mathbb B$ is the same in $m$ and $m^\prime$, whereas the total distance traveled by $\mathbb A$ is strictly less in $m$. Since $m^\prime$ was a lower bound for all counter-clockwise optimal motions, this implies that any clockwise optimal motion would be shorter than a counter-clockwise one. Thus we may restrict our attention to clockwise optimal motions only.
\end{proof}

The intermediate point, $A^\prime$ was not strictly necessary here as we could have also moved $A_0$ straight to $A_1$ on the first step. In constructing the lower bounds below however, we will make use of a judiciously chosen intermediate point.

In any case we shall encounter, the clockwise optimal motion is similar to counter-clockwise motions we've already considered. For this case, the optimal clockwise motion looks like a vertically reflected version of Figure~\ref{fig:special}. The intermediate pivot point $A_\text{int}$ is formed by using the intersection of lower tangents from the $A_i$'s to the $s\text{-circ}(B_j)$'s, where $i \neq j$. In general, we have the following lemma:

\begin{lemma}
\label{lem:optimal-clockwise}
Suppose $A_0, A_1 \in s\text{-corr}(B_0, B_1)$ and let $H_{ij}$ denote the half space below the upper tangent from $A_i$ to $s\text{-circ}(B_j)$. If $H_{ij}$ intersects $s\text{-circ}(B_i)$ for some $i\in\{0,1\}$, $j=1-i$, and $A_j \in H_{ij}$, then the optimal motion must be net clockwise.
\end{lemma}

\begin{proof}

There are two major cases: (i) the case where $s\text{-circ}(B_0)$ does not intersect $s\text{-circ}(B_1)$ and (ii) the case where they do intersect. For both cases, we assume that $A_0$ is under the line connecting $B_0$ with $B_1$. The other cases are treated similarly with almost exactly the same proof. 

\paragraph{$s\text{-circ}(B_0)$ does not intersect $s\text{-circ}(B_1)$}
Let $U_0$ be the upper tangent point of $A_0$ to $s\text{-circ}(B_1)$. By our assumptions, $A_1$ lies below $\overline{A_0U_0}$, and $A_1 \in s\text{-corr}(B_0, B_1)$. Let $U_1$ be the upper tangent point of $A_1$ to $s\text{-circ}(B_0)$. We first deal with the case where the tangent segments $\overline{A_0U_0}$ and $\overline{A_1U_1}$ intersect at a point $A_{\text{int}}\in s\text{-corr}(B_0, B_1)$ (see Figure \ref{fig:optimal-clockwise0}).

\begin{figure}[ht]

\begin{center}
\includegraphics[width=\columnwidth]{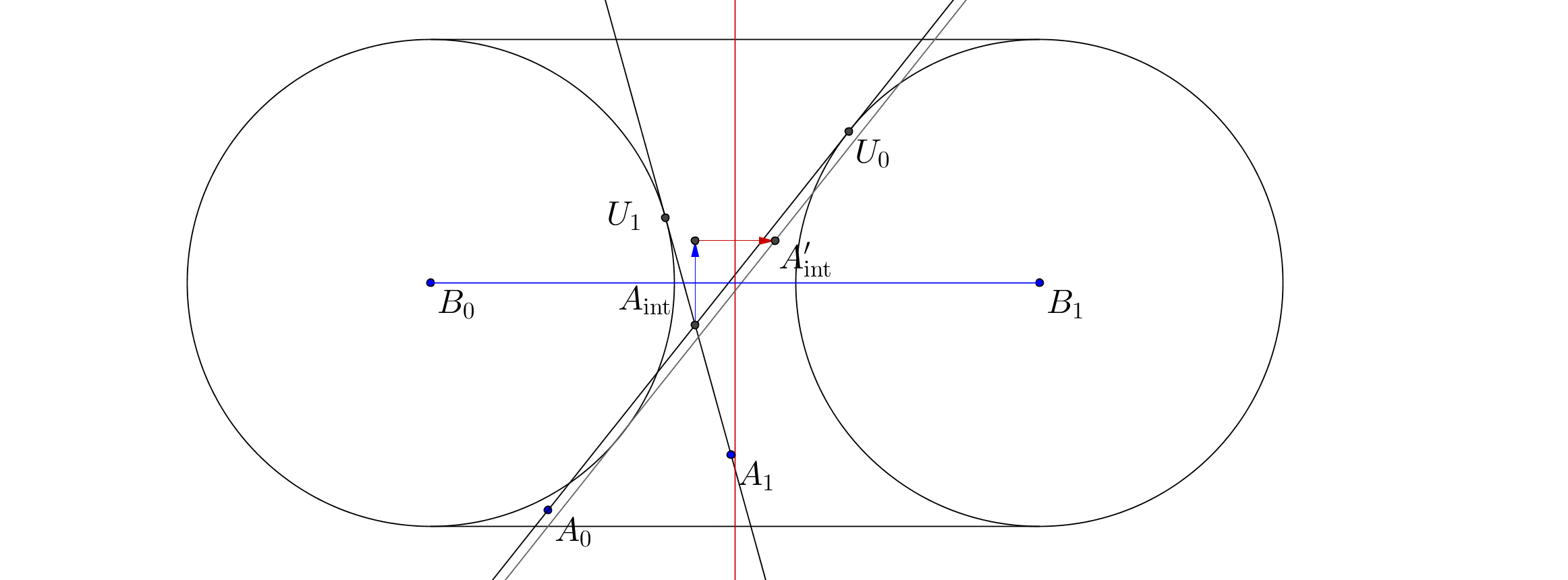}
\end{center}

\caption{ A case of Lemma \ref{lem:optimal-clockwise}. \label{fig:optimal-clockwise0}}
\end{figure} 

Consider the following ``motion'' $m^\prime = (\xi_{\mathbb A}^\prime, \xi_{\mathbb B}^\prime)$:
\begin{enumerate}
\item Move $\mathbb A$ on a straight line from $A_0$ to $A_{\text{int}}$.
\item Move $\mathbb B$ from $B_0$ to $B_1$ avoiding $s\text{-circ}(A_{\text{int}})$. This involves moving $\mathbb B$ to $T_{0}$ (the lower tangent point of $B_0$ and $s\text{-circ}(A_{\text{int}})$), rotating $\mathbb B$
counter-clockwise about $A_{\text{int}}$ to $T_{1}$ (the lower tangent point of $B_1$ and $s\text{-circ}(A_{\text{int}})$) in a range of angles $[\beta_{0},\beta_{1}]$, and then moving $\mathbb B$ from $T_{1}$ to $B_1$. 
\item Move $\mathbb A$ in a straight line from $A_{\text{int}}$ to $A_1$. 
\end{enumerate}

The ``motion'' outlined above is infeasible, as the position of $B_0$ prevents the movement from $A_0$ to $A_{\text{int}}$ in a straight-line. However, Lemma~\ref{lem:support} shows that $\ell(m^\prime)$ forms a lower bound on all possible net clockwise motions.

Now we construct a net clockwise motion whose length is no greater than that of $m^\prime$. Construct the point $A_{\text{int}}^{\prime}$ in Figure \ref{fig:optimal-clockwise0}, which is the result of two reflections of $A_{\text{int}}$, first along the line from $B_0$ to $B_1$ and then along the perpendicular bisector of $\overline{B_0B_1}$. Consider the following motion $m$:

\begin{figure}[ht]

\begin{center}
\includegraphics[width=\columnwidth]{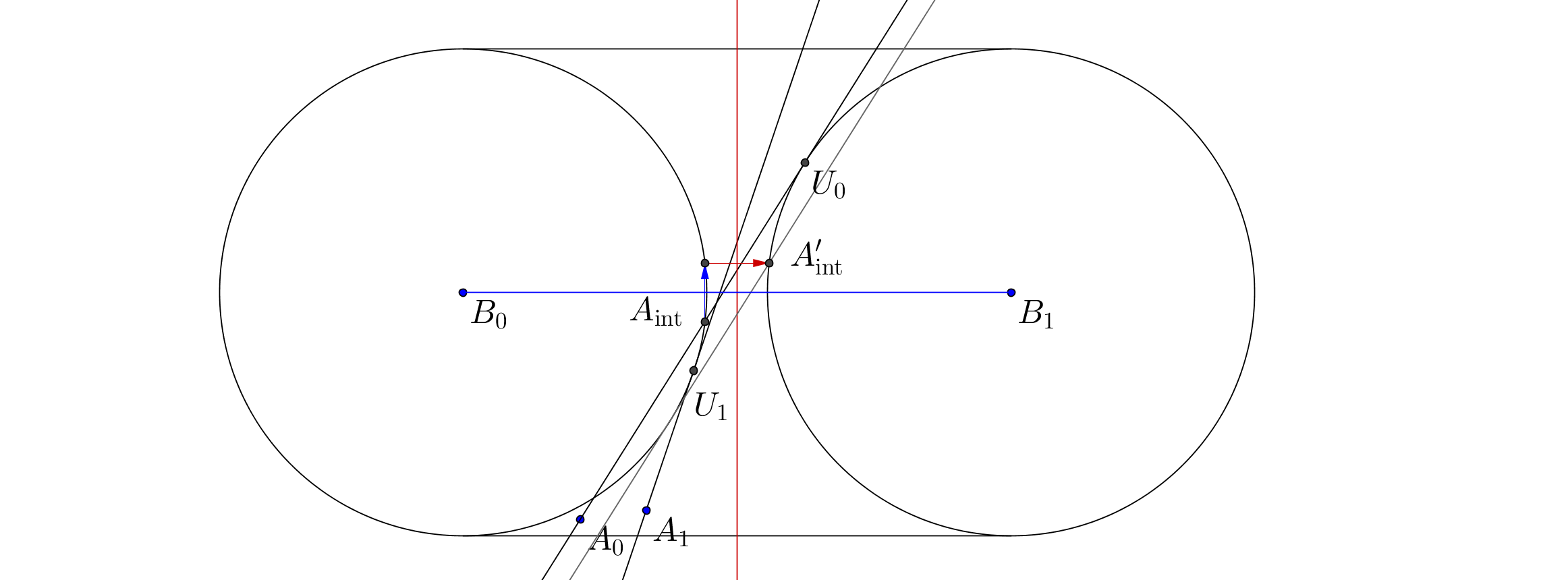}
\end{center}

\caption{ A case of Lemma \ref{lem:optimal-clockwise}. \label{fig:optimal-clockwise1}}
\end{figure} 

\begin{enumerate}
\item Move $\mathbb B$ from $B_0$ to $B_1$ avoiding $s\text{-circ}(A_{\text{int}}^{\prime})$ by rotating over the top of it.
\item Move $A_0$ to $A_1$ in a straight line.
\end{enumerate}

Clearly step 1 of $m$ is the same length as step 2 of $m^\prime$, and step 2 of $m$ is at most the length of steps 1 and 3 of $m^\prime$, so $\ell(m)\leq \ell(m^\prime)$. Furthermore $m$ is a feasible motion. To see this, let $t$ be line through $A_{\text{int}}^{\prime}$ parallel to the segment $\overline{A_0A_{\text{int}}}$, and let $q$ be the tangent point between $s\text{-circ}(B_0)$ and $t$. Note that $A_{\text{int}}^{\prime}$ lies on the right of $q$ above $t$ and $A_0$ is left of $q$ and above $t$, so $A_0$ does not obstruct the movement of $\mathbb B$ in step 1. 

Hence the optimal motion must be net clockwise in the case where $\overline{A_0U_0}$ and $\overline{A_1U_1}$ intersect.

When $\overline{A_0U_0}$ and $\overline{A_1U_1}$ do not intersect (see Figure \ref{fig:optimal-clockwise1}), this means that $U_1$ is below $\overline{A_0U_0}$. In this case, let $A_{\text{int}}$ be the right-most intersection point between $\overline{A_0U_0}$ and $s\text{-circ}(B_0)$ and the proof above will work without modification.

\paragraph{$s\text{-circ}(B_0)$ intersects $s\text{-circ}(B_1)$}
We now deal with case (ii), where $s\text{-circ}(B_0)$ intersects $s\text{-circ}(B_1)$ (see Figure~\ref{fig:optimal-clockwise2}). Let $\mathcal{L}$ (resp. $\mathcal{U}$) denote the region within $s\text{-corr}(B_0, B_1)$ below (resp. above) the discs enclosed by $s\text{-circ}(B_0)$ and $s\text{-circ}(B_1)$. We will show that if both $A_0$ and $A_1$ are in $\mathcal{L}$, then the optimal motion must be net clockwise. The case for $\mathcal{U}$ can be handled similarly.

As before, we will first lower bound the optimal net counter-clockwise motion by an infeasible motion, and then show a net clockwise motion that is at most the length of the lower bound. 

Let $t$ be the upper intersection point of $s\text{-circ}(B_0)$ and $s\text{-circ}(B_1)$. 
If $A_0$ is left of the perpendicular bisector of $\overline{B_0B_1}$, then define the following: $U_0$ is the upper tangent point of $A_0$ to $s\text{-circ}(B_1)$, $U_1$ is the upper tangent point of $A_1$ to $s\text{-circ}(B_0)$.  If $A_1$ is right of the perpendicular bisector, let $U_0$ be the upper tangent point of $A_0$ to $s\text{-circ}(B_0)$, and let $U_1$ be the upper tangent point of $A_1$ to $s\text{-circ}(B_1)$.

If $U_1$ is counter-clockwise of $t$ on $s\text{-circ}(B_0)$ or $U_0$ is clockwise of $t$ on $s\text{-circ}(B_1)$, one can check that the proof of the non-intersecting case works here as well. Otherwise, both $U_0$ and $V_1$ are vertically below $t$.

In this case, consider the following ``motion'' $m^\prime$:
\begin{enumerate}
\item Move $\mathbb A$ on a straight line from $A_0$ to $t$. This involves possibly moving on a chord through $s\text{-circ}(B_0)$ and $s\text{-circ}(B_1)$ in a range of angles $[\alpha_0, \alpha_1]$.
\item Move $\mathbb B$ from $B_0$ to $B_1$ avoiding $s\text{-circ}(t)$.
\item Move $\mathbb A$ in a straight line from $t$ to $A_1$. This involves possibly moving on a chord through $s\text{-circ}(B_0)$ and $s\text{-circ}(B_1)$ in a range of angles $[\alpha_2, \alpha_3]$.
\end{enumerate}

As in the previous case, Lemma~\ref{lem:support} (with $t$ as the dominating point) shows that $\ell(m^\prime)$ forms a lower bound on all net counter-clockwise motions.

%We note that the $m^\prime$ used here is almost the same as the $m^\prime$ used for the non-intersecting case, using $t$ in place of $A_{\text{int}}$. For the range of angles $[\alpha_0, \alpha_1]$ one support point will be $B_0$ as $[\alpha_0, \alpha_1]\subset [-\pi/2, \pi/2]$. The other support point will either be $t$ or one of the $A_i$'s. Note that when the other support point is $t$, the support function at most $s$ (and strictly less at all but 1 point) as $m^\prime$ cuts through $s\text{-circ}(B_0)$. Hence $h_{\mathbb{AB}}(\theta)$ is at most its lower bound for $\theta \in [\alpha_2, \alpha_3]$ Similarly, for the range of angles $[\alpha_2, \alpha_3]$, one support point will be $B_1$, and the other support point will either be $t$ or one of the $A_i$'s. Again, a similar argument shows that $h_{\mathbb{AB}}(\theta)$ meets its lower bound in $[\alpha_2, \alpha_3]$ as well.

%For all other angles, the argument proceeds exactly as in the non-intersecting case.

Now we construct a net clockwise motion whose length is no greater than that of $m^\prime$. Construct the point $A_\text{int}^\prime$, which is the vertical reflection of $t$ across $\overline{B_0B_1}$. Now consider the same type of motion $m$ that we used in the non-intersecting case:

\begin{enumerate}
\item Move $\mathbb B$ from $B_0$ to $B_1$ avoiding $s\text{-circ}(A_{\text{int}}^{\prime})$ by rotating over the top of it.
\item Move $A_0$ to $A_1$ in a straight line.
\end{enumerate}

\begin{figure}[ht]

\begin{center}
\includegraphics[width=0.8\columnwidth]{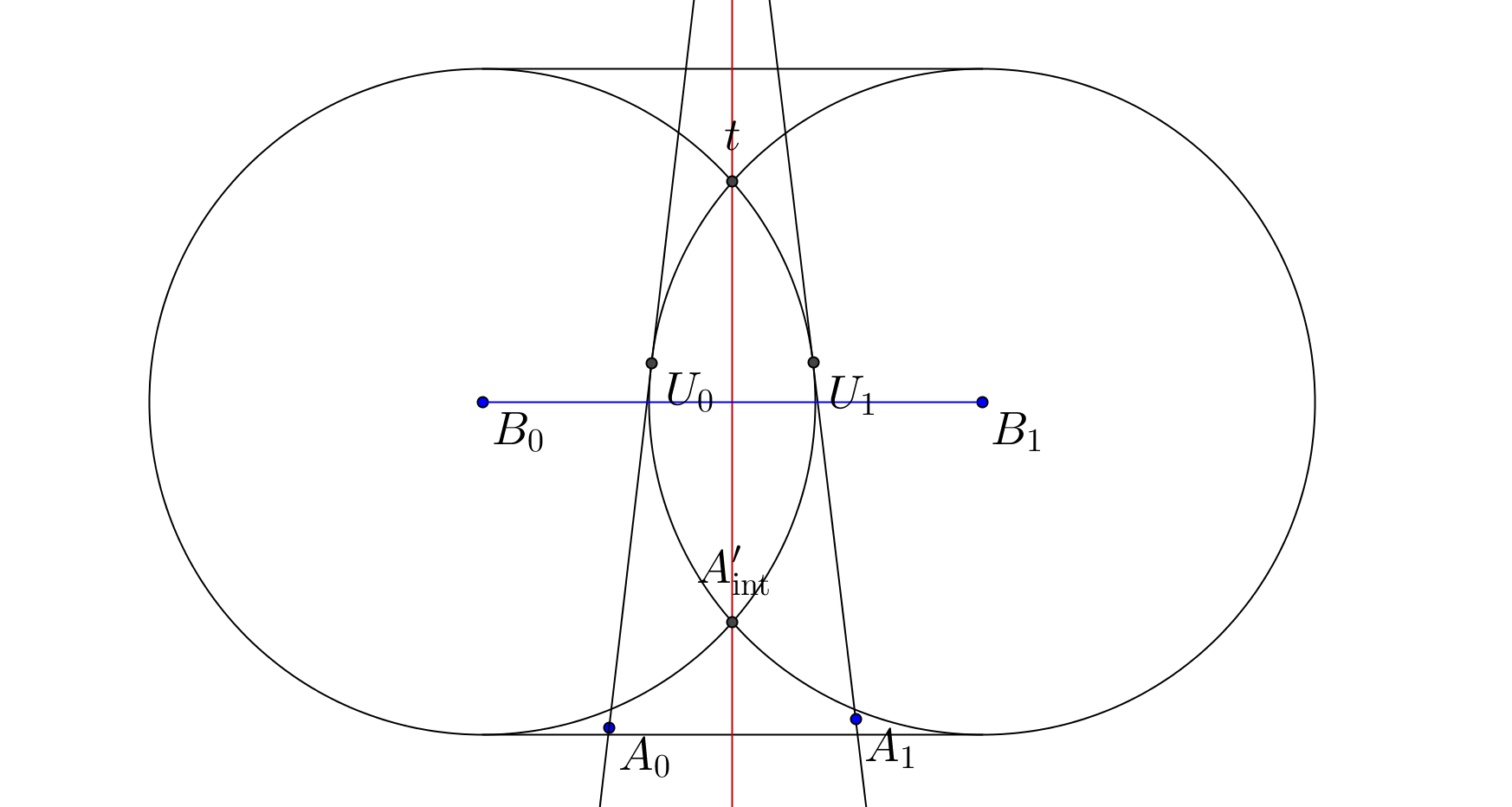}
\end{center}

\caption{ A case of Lemma \ref{lem:optimal-clockwise}. \label{fig:optimal-clockwise2}}
\end{figure} 

Clearly $m$ is a feasible motion. As before, step 1 of $m$ is the same length as step 2 of $m^\prime$, and step 2 of $m$ is at most the length of steps 1 and 3 of $m^\prime$, so $\ell(m)\leq \ell(m^\prime)$.
\end{proof}

%%%%%%%%%%%%%%%%%%%%%%%%%%%%%
\section{Case analysis of counter-clockwise optimal motions}
\label{sec:mainprf2}
In this section we treat exhaustively each case of Table~\ref{straight-line}, beginning with Case 1. For Case 2 and onwards, the general form of the motion we construct will be similar to examples presented in Section~\ref{sec:mainprf}. That is, the motion will be decoupled, consisting of at most two $\mathbb A$ motions which meet at an intermediate point $A_\text{int}$ and one $\mathbb B$ motion. The motions themselves are constructed from tangent segments and arcs of radius $s$ circles. When an arc of a circle is part of a motion, the centre of the circle will be dominating in the sense of Definition~\ref{def:dominating}.

\subsection{Case 1}
It suffices to treat Case 1a, as Case 1b reduces to Case 1a by symmetry. %either symmetry (i) (interchanging the initial and final placements) or (ii) (switching the roles of $\mathbb A$ and $\mathbb B$).
 In Case 1a, $A_0 \not\in s\text{-corr}(B_0, B_1)$, so on the first step we translate $\mathbb B$ from $B_0$ to $B_1$ in a straight line without touching $\mathbb A$. At this point $\mathbb A$ can move freely in a straight line from $A_0$ to $A_1$, as $B_1 \not\in s\text{-corr}(A_0, A_1)$. As we shall see through examining the other cases, Case 1 is the only situation where a straight-line motion is possible.

\subsection{Case 2}
\label{sec:newcase2}
It suffices to treat Case 2a since Case 2b reduces to 2a by symmetry; % (i); 
thus we assume that $A_0 \in s\text{-corr}(B_0, B_1)$ and $B_0 \in s\text{-corr}(A_0, A_1)$. In fact, we can relax this and assume that $A_1 \in s\text{-cone}(A_0, B_0)$. This amounts to including the ``wedge'' between $A_0$ and $s\text{-circ}(B_0)$.

\begin{figure}[ht]

\begin{subfigure}{\textwidth}
  \begin{center}
  \includegraphics[width=\textwidth]{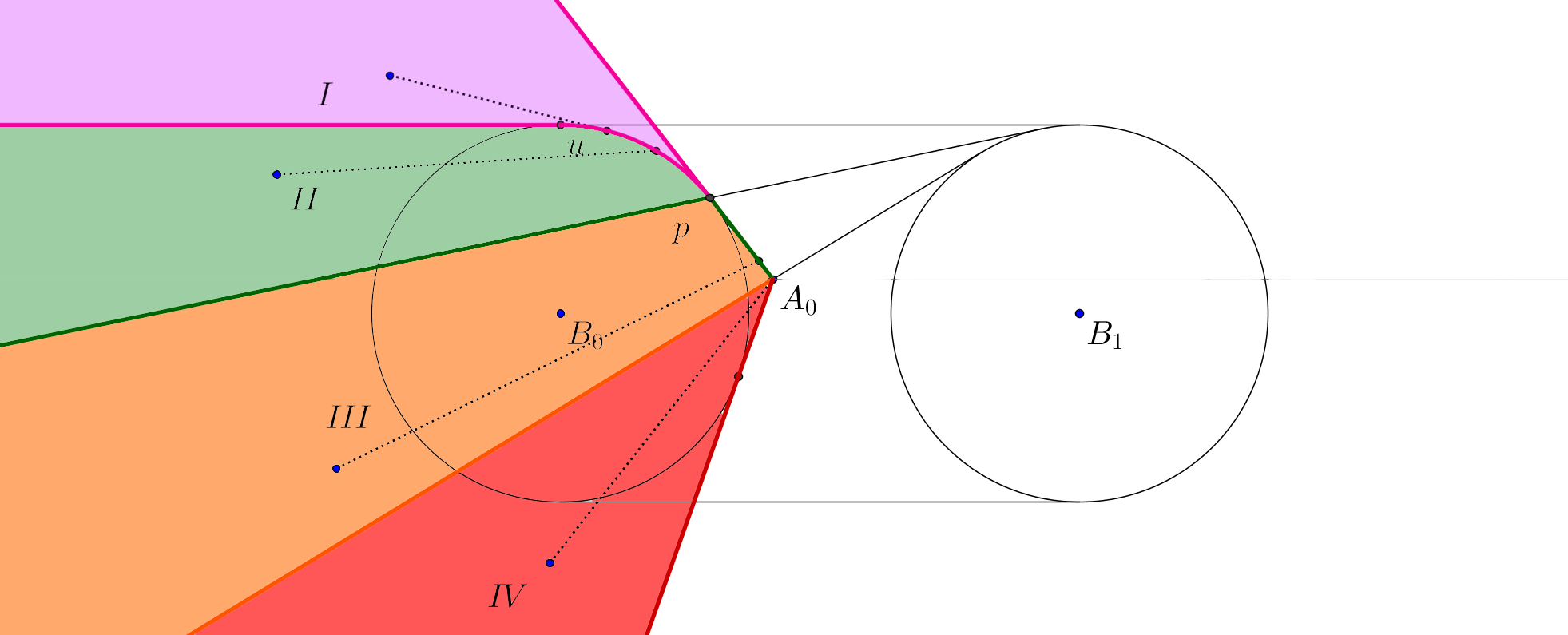}
  \end{center}
  \caption{\label{fig:case2a-separated}}
\end{subfigure}

%\end{figure}

%\begin{figure}[ht]
%\ContinuedFloat
\begin{subfigure}{\textwidth}
  \begin{center}
  \includegraphics[width=\textwidth]{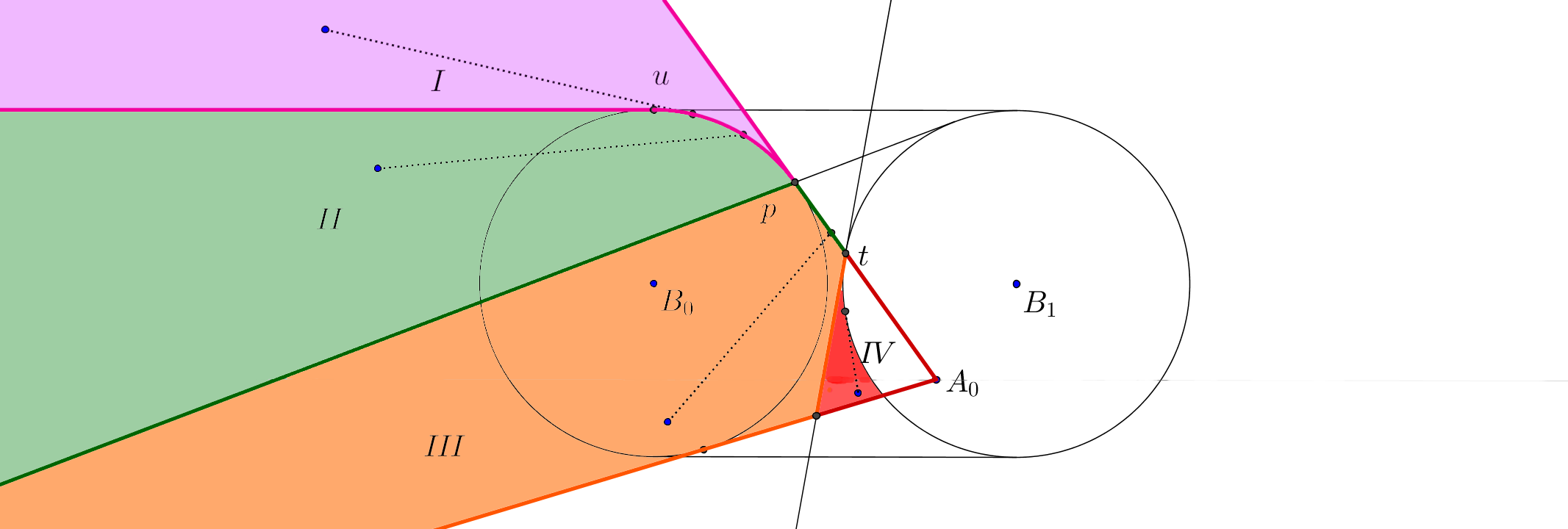}
  \end{center}
  \caption{\label{fig:case2a-separated-a0-in-b1}}
\end{subfigure}

\caption{The different zones of Case 2 when $s\text{-circ}(B_0)$ and $s\text{-circ}(B_1)$ do not intersect. We have different optimal motions (dotted lines) depending on the zone in which $A_1$ lies. \label{case2a}}

\end{figure}
The motion we take in Case 2a depends on the zone in which $A_1$ lies (cf. Figure \ref{case2a} and \ref{case2a-intersect}). Each zone represents a locus of locations for $A_0$ which give rise to a specific sequence of motions that are counter-clockwise optimal within that zone.

Let $p$ be the upper tangent point from $A_0$ to $s\text{-circ}(B_0)$. The zones are defined by the following properties:

\begin{itemize}
\itemsep0em
\item [Zone I:] The set of points $q \in s\text{-cone}(A_0, B_0)$ for which some tangent point from $q$ to $s\text{-circ}(B_0)$ lies on the arc of $s\text{-circ}(B_0)$ from $p$ to $u$.
\item [Zone II:] The set of points $q \in s\text{-cone}(A_0, B_0)$ where the tangent from $q$ to $s\text{-circ}(B_1)$ intersects the arc of $s\text{-circ}(B_0)$ from $p$ to $u$.
\item [Zone III:] The set of points $q \in s\text{-cone}(A_0, B_0)$ where the tangent from $q$ to $s\text{-circ}(B_1)$ intersects $\overline{A_0p}$.
\item [Zone IV:] The set of points $q \in s\text{-cone}(A_0, B_0)$ that are dominated by $t$. $t$ is $A_0$ if $A_0\not\in s\text{-circ}(B_1)$, is the intersection point of $\overline{A_0p}$ and $s\text{-circ}(B_1)$ if $A_0 \in s\text{-circ}(B_1)$, and is the upper intersection point of $s\text{-circ}(B_0)$ and $s\text{-circ}(B_1)$ if the intersection point of the circles lie on the arc from $p$ to $u$.
\end{itemize}

For concreteness, we also give constructive definitions in each subcase below.

%%%%%%%%%%%%%%%% %%%%%%%%%%%%

\subsubsection{Subcase 1: $s\text{-circ}(B_0)$ and $s\text{-circ}(B_1)$ do not intersect}

We first discuss the constructions of zones I-IV in Figures \ref{fig:case2a-separated} and \ref{fig:case2a-separated-a0-in-b1}. We may construct zones I-IV explicitly through the following tangents and curves: 

\begin{enumerate}
\item The horizontal tangent through the uppermost point $u$ of $s\text{-circ}(B_0)$. This tangent and the arc of $s\text{-circ}(B_0)$ between $u$ and $p$ (where $p$ is the upper tangent point between $A_0$ and $s\text{-circ}(B_0)$) separates zone I from zone II.
\item The tangent through $p$ to $s\text{-circ}(B_1)$. This tangent separates zone II from zone III.
\item If $A_0\not\in s\text{-circ}(B_1)$, the tangent line from $A_0$ to $s\text{-circ}(B_1)$ (cf. Figure \ref{fig:case2a-separated}). Otherwise, the tangent of $s\text{-circ}(B_1)$ through $t$, where $t$ is the intersection point of $\overline{A_0p}$ and $s\text{-circ}(B_1)$ (cf. Figure \ref{fig:case2a-separated-a0-in-b1}). This tangent separates zone III from zone IV.
\end{enumerate}

Note that zone III and IV may be empty, if the position of $A_0$ lies below the line tangent to the bottom of $\text{circ}_s(B_0)$ and the top of $\text{circ}_s(B_1)$. 

For each zone we specify the location of the intermediate point $A_{\text{int}}$ as follows: 

\begin{itemize}
\itemsep0em
\item [Zone I:] $A_{\text{int}}$ is the point $A_1$.

\item [Zone II:] $A_{\text{int}}$ is the rightmost point of intersection between the tangent from $A_1$ to $\text{circ}_s(B_1)$ and $\text{circ}_s(B_0)$.

\item [Zone III:] $A_{\text{int}}$ the point of intersection of the tangent from $A_1$ to $\text{circ}_s(B_1)$ and the tangent from $A_0$ to
$\text{circ}_s(B_0)$. 

\item [Zone IV:] $A_{\text{int}}$ is the point $t$ (as defined above).
\end{itemize}

We define points $T_0$ and $T_1$ which are the lower points of tangency to $\text{circ}_s(A_{\text{int}})$ from $B_0$ and $B_1$ respectively.
Our three-step generic motion involves:
\begin{enumerate}
\item Moving $\mathbb A$ on the shortest path from $A_0$ to $A_{\text{int}}$, avoiding $\text{circ}_s(B_0)$. This may involve rotating $\mathbb A$ counter-clockwise
about $B_0$ in a range of angles $[\alpha_{0},\alpha_{1}]$. 
\item Moving $\mathbb B$ from $B_0$ to $B_1$ avoiding $\text{circ}_s(A_{\text{int}})$. This involves translating $\mathbb B$ from $B_0$ to $T_0$, rotating $\mathbb B$
counter-clockwise about $A_{\text{int}}$ from $T_0$ to $T_{1}$ in a range of angles $[\beta_{0},\beta_{1}]$,
and then translating $\mathbb B$ from $T_{1}$ to $B_1$.
\item Translating $\mathbb A$ from $A_{\text{int}}$ to $A_1$ (collision-free by the disjointness of $\text{cone}_s(A_0, B_0)$ and $\text{circ}_s(B_1)$). 
\end{enumerate}

From the descriptions above, one can see that there is some amount of symmetry between zone I and IV. For this reason, we first dispense with Zones II and III, and then handle Zone I and IV at the end of this section.

\subsubsection*{$A_1$ is in zone II}

If $A_1$ is in zone II, then the tangent from $A_1$ to $B_1$
must intersect $B_0$ in up to two points. Let $A_{\text{int}}$ be the rightmost
intersection point.

\begin{proof}
Since $B_0$ dominates $B_1$ with respect to $s\text{-corr}(A_0, A_1)$, we have by Lemma~\ref{lem:support} that $h_{\mathbb{AB}}(\theta)=s$ for $\theta\in[\alpha_0, \alpha_1]$. Similarly, since $A_\text{int}$ dominates $A_0$ and $A_1$ with respect to $s\text{-corr}(B_0, B_1)$, $h_{\mathbb{AB}}(\alpha)=s$ for $\theta\in[\beta_0, \beta_1]$ by Lemma~\ref{lem:support}.

For angles in $S^1-[\alpha_{0},\alpha_{1}]-[\beta_{0},\beta_{1}]$, one can check that $A_0$ or $A_1$ must be one support point, and either $B_0$ or $B_1$ must be the other.
\end{proof}

\subsubsection*{$A_1$ is in zone III}

\begin{proof}
By construction, $A_\text{int}$ dominates $A_0$ and $A_1$ with respect to $s\text{-corr}(B_0, B_1)$. Hence by Lemma~\ref{lem:support}, we have $h_{\mathbb{AB}}(\theta)=s$ for $\theta \in [\alpha_0, \alpha_1]$.
For angles in $S^1-[\alpha_{0},\alpha_{1}]$, one can see that either $A_0$ or $A_1$ must be one support point, and either $B_0$ or $B_1$ must be the other.
\end{proof}

\subsubsection*{$A_1$ is in zone I}
There are two cases for Zone I, the location of $A_1$ with respect to the upper tangent $A_0$ and $s\text{-circ}(B_1)$. Let $U$ be the upper tangent of $A_0$ and $s\text{-circ}(B_1)$.

\begin{figure}[ht]

\begin{center}
  \includegraphics[width=0.7\textwidth]{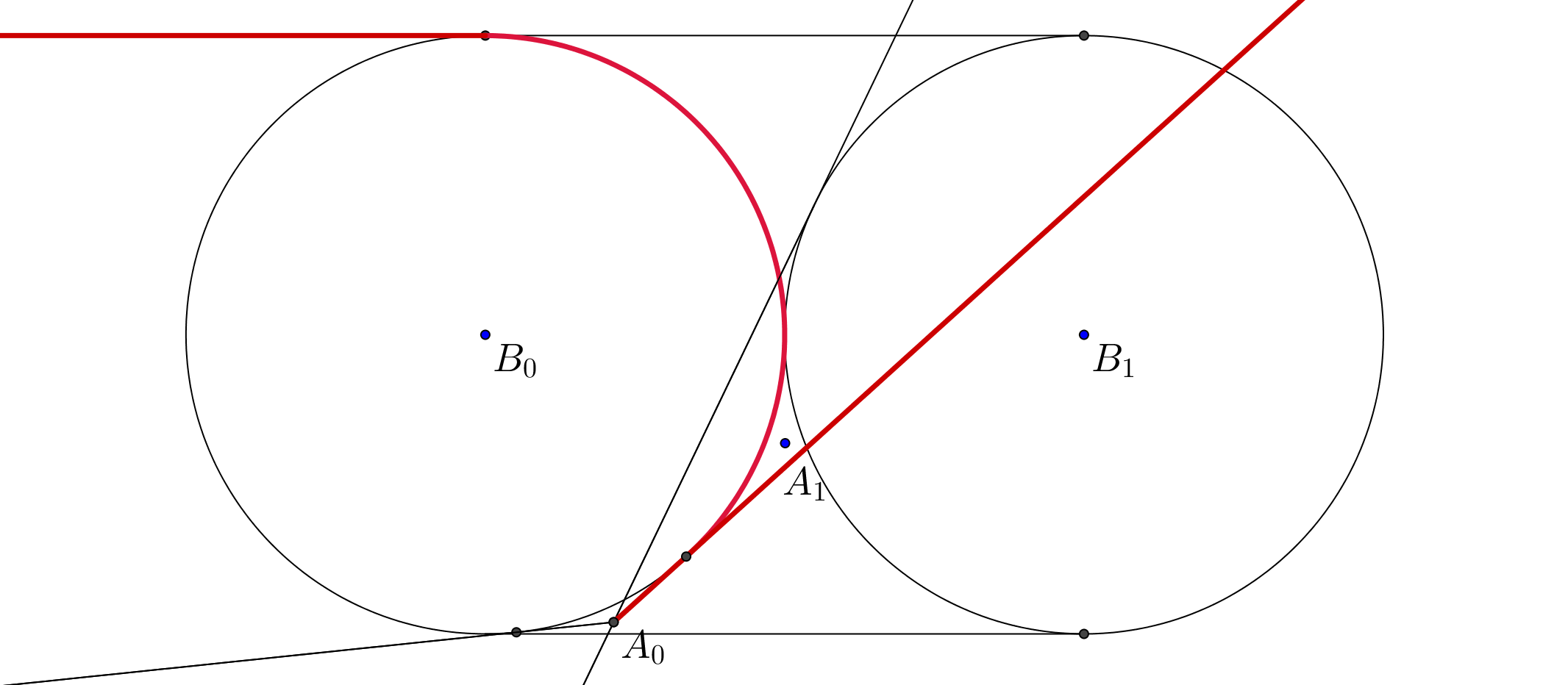}
  \end{center}
\caption{ Zone I, example of when $A_1$ is below $U$. Zone I is outlined by the bolded tangents. \label{fig:zone1subcase}}
\end{figure} 

\paragraph{$A_1$ is above $U$.}
\begin{proof}
In this case, $A_1$ either dominates $A_0$ or is outside of $s\text{-corr}(B_0, B_1)$ and so by Lemma~\ref{lem:support} choosing $A_1$ as $A_\text{int}$ shows that $h_{\mathbb{AB}}(\theta)=s$ for $\theta\in[\beta_0, \beta_1]$ (where $[\beta_0, \beta_1]=\emptyset$ for $A_1\not\in s\text{-corr}(B_0, B_1)$). Furthermore, $B_0$ dominates $B_1$ with respect to $s\text{-corr}(A_0, A_1)$, so Lemma~\ref{lem:support} again shows that $h_{\mathbb{AB}}(\theta)=s$ for $\theta\in[\beta_0, \beta_1]$. Since there are no intermediate pivot points except for the $A_i$'s and $B_i$'s, it's clear that for all other angles, $A_0$ or $A_1$ must be one support and $B_0$ or $B_1$ must be the other.
\end{proof}

\paragraph{$A_1$ is below $U$.}
In this case (see Figure~\ref{fig:zone1subcase}), the positions of $A_0$ and $A_1$ satisfy the conditions of Lemma~\ref{lem:optimal-clockwise}. Thus we may look for a clockwise motion. In the clockwise zones, $A_1$ is in Zone IV of $A_0$, which we handle below.

\subsubsection*{$A_1$ is in zone IV} 

Due to the complexity of Zone IV, we split it into three subcases. 

\paragraph{Zone IV, subcase 1.}

\begin{figure}[ht]

\begin{center}
  \includegraphics[width=\textwidth]{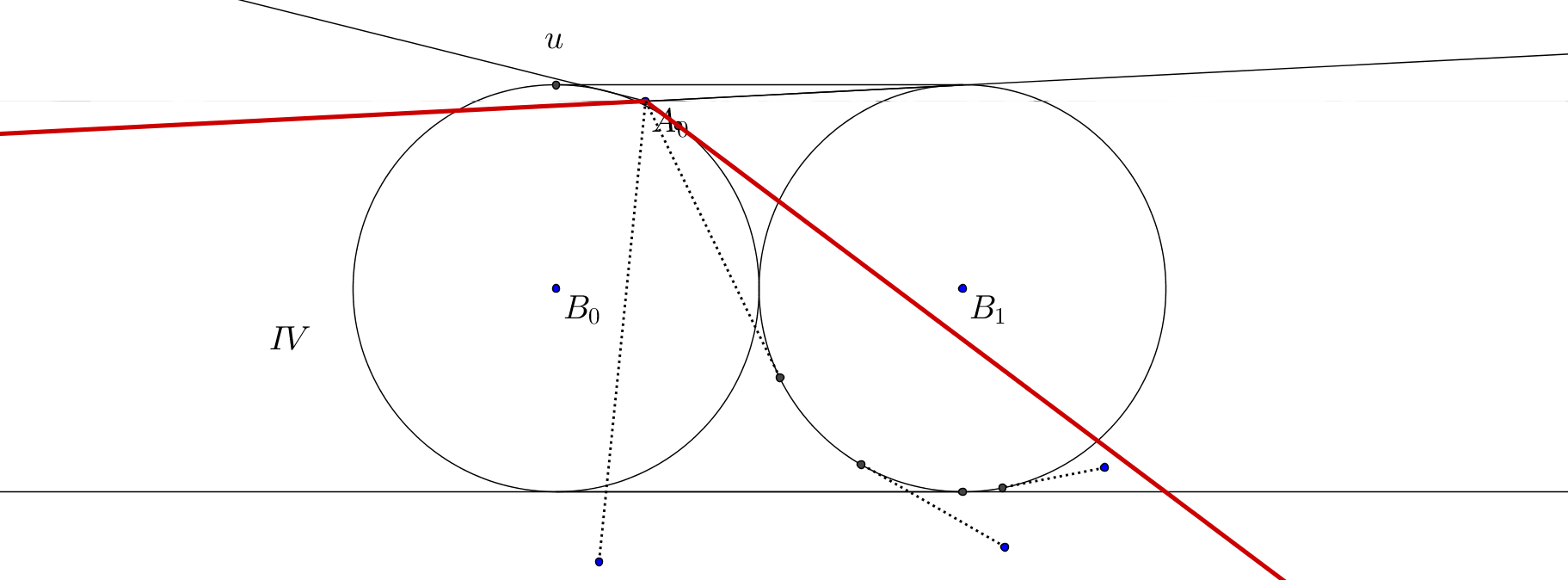}
  \end{center}
\caption{ Positions of $A_1$ in Zone IV, subcase 1 and subcase 2. Zone IV is outlined by the bolded tangents. \label{fig:zone4case1+2}}
\end{figure} 

We first handle the cases for which $A_0 \not\in s\text{-circ}(B_1)$ and the upper tangent point of $A_1$ and $s\text{-circ}(B_1)$ lies inside $s\text{-corr}(B_0,B_1)$. By these assumptions, we must have $A_1 \in s\text{-corr}(B_0, B_1)$ or below the lower horizontal tangent of $s\text{-circ}(B_0)$ and $s\text{-circ}(B_1)$ (see Figure \ref{fig:zone4case1+2}).

In this case, choosing $A_\text{int}$ to be $A_0$ in our three-step generic motion yields a net optimal counter-clockwise motion.
\begin{proof}
By construction of Zone IV, $A_0$ dominates $A_1$ with respect to $s\text{-corr}(B_0, B_1)$. Hence by Lemma~\ref{lem:support}, $h_{\mathbb{AB}}(\alpha)=s$ for $\alpha\in[\alpha_{0},\alpha_{1}]$. 

By our property that the upper tangent point of $A_1$ and $s\text{-circ}(B_1)$ lies inside $s\text{-corr}(B_0,B_1)$, we have that $B_1$ dominates $B_0$ with respect to $s\text{-corr}(A_0, A_1)$.

For angles in $S^1-[\alpha_{0},\alpha_{1}]-[\beta_{0},\beta_{1}]$, either $A_0$ or $A_1$ must be one support point, and either $B_0$ or $B_1$ must be the other. This is due to the fact that all pivot points in our motion are either the initial or final positions, and all non-pivots were either circular arcs or tangents.
\end{proof}

\paragraph{Zone IV, subcase 2.}
If subcase 1 does not apply and $A_0 \not\in s\text{-circ}(B_1)$, then $A_1$ must be right of the lower tangent between $A_0$ and $s\text{-circ}(B_1)$, and above or on the lower horizontal tangent of $s\text{-circ}(B_0)$ and $s\text{-circ}(B_1)$. See Figure \ref{fig:zone4case1+2} for an illustration of the possible positions of $A_1$ under our assumption.

In this case, Lemma~\ref{lem:optimal-clockwise} shows that there exists a net-clockwise motion that is at least as good as any net-counter-clockwise optimal motion, and that the net-clockwise optimal motion is simply the clockwise version of subcase 1 handled above. %In fact, we can show that net-clockwise motions are optimal in a wider range of cases. This is presented in Lemma \ref{lem:optimal-clockwise} below.

To apply Lemma~\ref{lem:optimal-clockwise}, we first rotate Figure \ref{fig:zone4case1+2} so that the $A_i$'s are on the $x$-axis (cf. Figure \ref{fig:zone4hard-rotated}). Next, let $V_0$ and $V_1$ be the upper tangent points of $B_0$ to $s\text{-circ}(A_1)$ and $B_1$ to $s\text{-circ}(A_0)$ respectively. As $A_1$ is above the lower horizontal tangent of $s\text{-circ}(B_0)$ and $s\text{-circ}(B_1)$, we must have $B_1$ under $\overline{B_0 V_1}$. Hence we may apply Lemma \ref{lem:optimal-clockwise} with the roles of $\mathbb A$ and $\mathbb B$ switched.

\begin{figure}[ht]

\begin{center}
\includegraphics[width=\textwidth]{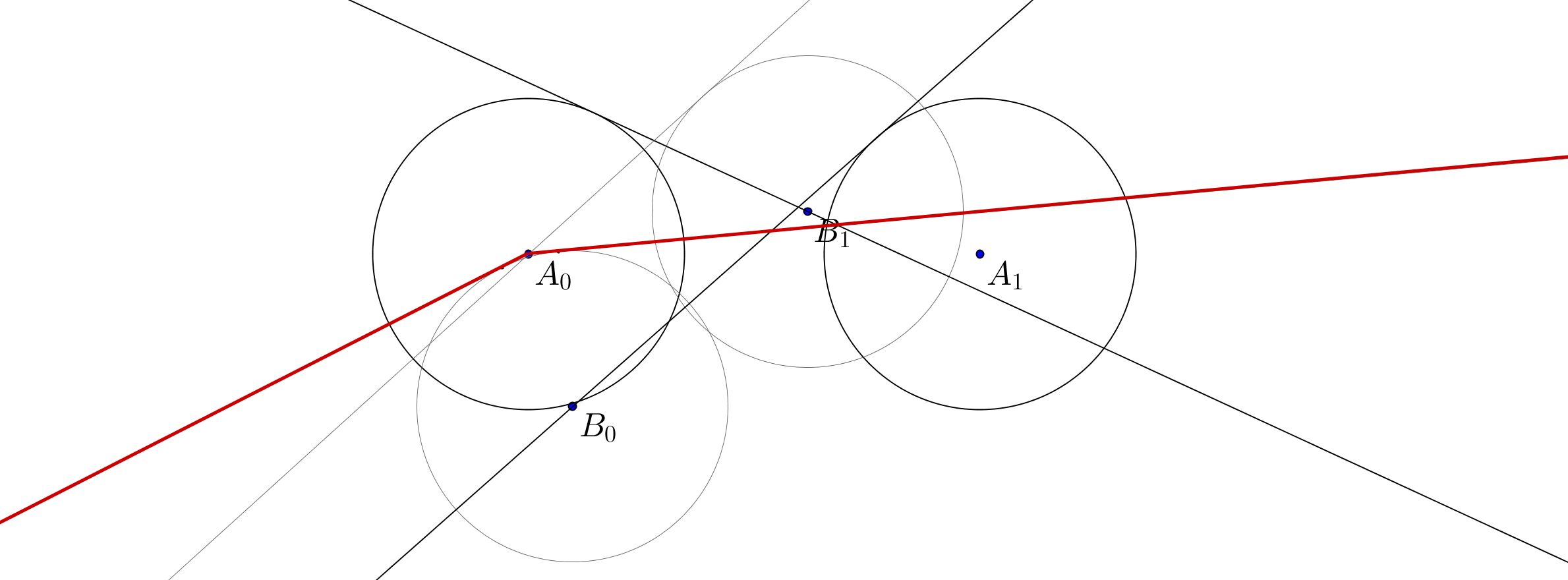}
\end{center}

\caption{ Zone IV subcase 2, rotated. \label{fig:zone4hard-rotated}}
\end{figure} 

The clockwise optimal motion is:
\begin{enumerate}
\item Move $\mathbb A$ from $A_0$ to $A_1$ rotating over the top of $s\text{-circ}(B_0)$. 
\item Move $\mathbb B$ in a straight line from $B_0$ to $B_1$. 
\end{enumerate}

This is exactly the motion in Zone IV, subcase 1, with the roles of $A$ and $B$ switched, so the optimality of this motion is already shown above. 

\paragraph{Zone IV, subcase 3.}

If subcase 1 and subcase 2 do not apply, then $A_0 \in s\text{-circ}(B_1)$ (cf. Figure \ref{fig:case2a-separated-a0-in-b1}).

In this case the optimal option is:
\begin{enumerate}
\item Move $\mathbb A$ on a straight line from $A_0$ to $t$.
\item Move $\mathbb B$ from $B_0$ to $B_1$ avoiding $s\text{-circ}(t)$. This involves moving $\mathbb B$ to $T_{0}$, rotating $\mathbb B$ counter-clockwise about $A_{\text{int}}$ to $T_{1}$ in a range of angles $[\alpha_{0},\alpha_{1}]$,
and then moving $\mathbb B$ from $T_{1}$ to $B_1$.
\item Move $\mathbb A$ on a shortest path from $t$ to $A_1$ while avoiding $s\text{-circ}(B_1)$. This involves rotating possibly rotating $\mathbb A$ in a range of angles $[\beta_0, \beta_1]$ around $s\text{-circ}(B_1)$.
\end{enumerate}

In this case, the motion is of the same type as the one given for Zone II and the exact same proof applies.

%%%%%%%%%%%%%%%% %%%%%%%%%%%%

\subsubsection{Subcase 2: $s\text{-circ}(B_0)$ and $s\text{-circ}(B_1)$ intersects}

\begin{figure}[ht]

\begin{subfigure}[b]{\textwidth}
\begin{center}
  \includegraphics[width=\textwidth]{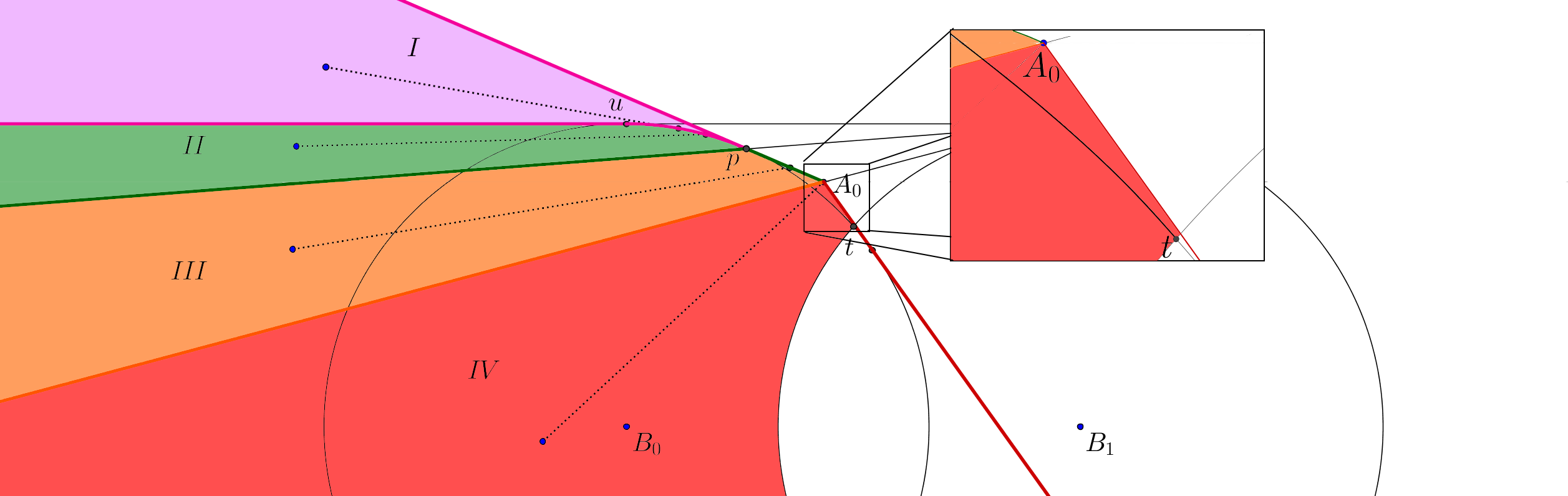}
  \end{center}
  \caption{ \label{case2a-intersect-upper}}
\end{subfigure}

%\end{figure}

%\begin{figure}[ht]
%\ContinuedFloat % break figure onto two pages
\begin{subfigure}[b]{\textwidth}
  \begin{center}
  \includegraphics{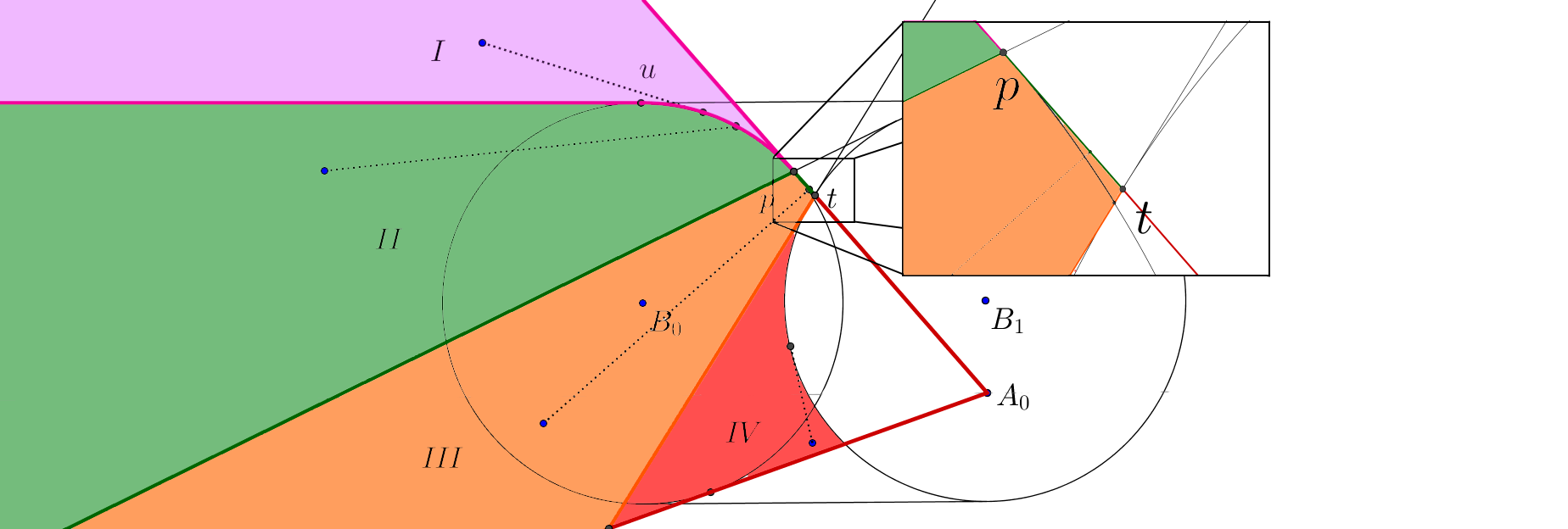}
  \end{center}
  \caption{}
  \label{case2a-intersect-lower1}
\end{subfigure}

\begin{subfigure}[b]{\textwidth}
  \begin{center}
  \includegraphics[width=0.95\textwidth]{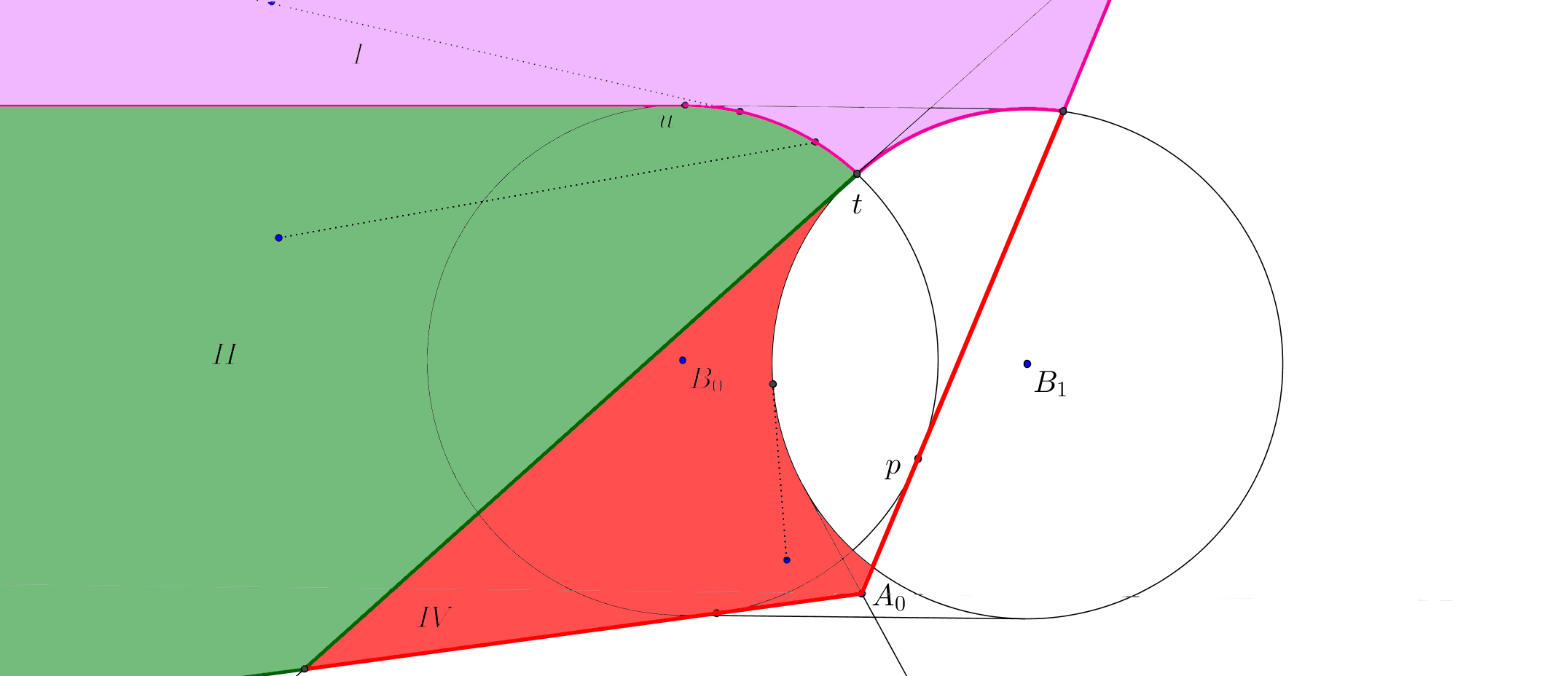}
  \end{center}
  \caption{}
  \label{case2a-intersect-lower2}
\end{subfigure}

\caption{The different zones of Case 2 when $s\text{-circ}(B_0)$ and $s\text{-circ}(B_1)$ intersect. We have different optimal motions (dotted lines) depending on the zone in which $A_1$ lies. \label{case2a-intersect}}

\end{figure}

When $s\text{-circ}(B_0)$ and $s\text{-circ}(B_1)$ intersect (cf. Figure \ref{case2a-intersect}), the zones are defined by the following curves:
 
\begin{enumerate}
\item The two tangents from $A_0$ to $s\text{-circ}(B_0)$.

\item The horizontal tangent from the top of $s\text{-circ}(B_0)$.

\item The tangent from $p$ to $s\text{-circ}(B_1)$ where $p$ is the
upper tangent point from $A_0$ to $s\text{-circ}(B_0)$. 

\item The tangent line from $t$ to $s\text{-circ}(B_1)$. Let $v$
be the intersection point of the line from $A_0$ to $p$ and $s\text{-circ}(B_1)$. If $A_0 \not\in s\text{-circ}(B_1)$, then $t$ is $A_0$. Otherwise, $A_0 \in s\text{-circ}(B_1)$ and $t$ is $v$ when $p$ lies outside of the 
$s\text{-circ}(B_1)$ (Figure \ref{case2a-intersect-lower1}), and $t$ is the upper intersection point between the $s\text{-circ}(B_i)$'s otherwise (Figure \ref{case2a-intersect-lower2}). 

\end{enumerate}

For the most part, the motions executed in Subcase 1 and Subcase 2 are the same, as are their intermediate points. However, for Zone I and IV there are small differences, as we shall see.

\subsubsection*{$A_1$ is in zone II or III}
In these zones, the motion is the same as the non-intersecting case.

\subsubsection*{$A_1$ is in zone I}
By Lemma~\ref{lem:optimal-clockwise}, if $A_1$ is located in any portion of Zone I which intersects the region below the $s\text{-circ}(B_i)$'s, then the motion must be net-clockwise optimal (an example can be found in Figure~\ref{fig:zone1subcase}, with the $B_i$'s pushed closer together). In this case, $A_1$ is in Zone IV of the clockwise zones, which we handle below. Otherwise, the motion for Zone I is the same as in Subcase 1, and the same proof applies.

\subsubsection*{$A_1$ is in zone IV}
Here we divide the motion into two different cases, depending on whether we are in Figure \ref{case2a-intersect-upper}, or \ref{case2a-intersect-lower1} and Figure \ref{case2a-intersect-lower2}. To be precise, denote $\mathcal{U}$ to be the region of $s\text{-corr}(B_0, B_1)$ that is above the $s\text{-circ}(B_i)$'s. We divide into two cases, depending on whether $A_0 \in \mathcal{U}$ or not.

\paragraph{Zone IV, subcase 1. $A_0 \in \mathcal{U}$}
In this case, the motions are exactly the same as those for Zone IV of the non-intersecting case.

\paragraph{Zone IV, subcase 2. $A_0 \not\in \mathcal{U}$}
This case is shown in Figures \ref{case2a-intersect-lower1} and \ref{case2a-intersect-lower2}. First, if $A_1$ is right of the upper tangent between $A_0$ and $s\text{-circ}(B_1)$ and left of the upper tangent between $A_0$ and $s\text{-circ}(B_0)$, then Lemma~\ref{lem:optimal-clockwise} shows that the optimal motion must be clockwise. In this case, the optimal clockwise motion is:
\begin{enumerate}
\item Move $\mathbb B$ from $B_0$ to $B_1$ rotating over the top of $s\text{-circ}(A_0)$. 
\item Move $\mathbb A$ in a straight line from $A_0$ to $A_1$. 
\end{enumerate}

\begin{proof}
The optimality of this motion can be see by reflecting the configuration vertically. Since $A_0$ dominates $A_1$, Lemma~\ref{lem:support} shows that $h_\mathbb{AB}(\theta)=s$ in the angles of rotation. For all other angles, the two support points are either $A_0$ or $A_1$ and $B_0$ or $B_1$.
\end{proof}

Now we assume that $A_1$ is outside of the region handled above. Let $T_{0}$ and $T_{1}$ be the lower tangent points of $B_0$
and $B_1$ to $s\text{-circ}(t)$ respectively. Let $V_0$ be the upper tangent point between $t$ and $s\text{-circ}(B_0)$. In this case the optimal motion is:
\begin{enumerate}
\item Move $\mathbb A$ on a shortest path from $A_0$ to $t$ while avoiding $s\text{-circ}(B_0)$. If $V_0 \not\in s\text{-circ}(B_1)$, this is simply a straight line and we define $[\beta_{0},\beta_{1}]=\emptyset$. Otherwise, this involves moving $\mathbb A$ to $V_0$, and rotating $\mathbb A$ in a range of angles $[\gamma_{0},\gamma_{1}]$ from $V_0$ to $t$. 
\item Move $\mathbb B$ from $B_0$ to $B_1$ avoiding $s\text{-circ}(t)$. This involves moving $\mathbb B$ to $T_{0}$, rotating $\mathbb B$
counter-clockwise about $A_{\text{int}}$ to $T_{1}$ in a range of angles $[\alpha_{0},\alpha_{1}]$,
and then moving $\mathbb B$ from $T_{1}$ to $B_1$.
\item Move $\mathbb A$ on a shortest path from $t$ to $A_1$ while avoiding $s\text{-circ}(B_1)$. This involves rotating $\mathbb A$ in a range of angles $[\gamma_2, \gamma_3]$ around $s\text{-circ}(B_1)$.
\end{enumerate}

\begin{proof}
The optimality of this motion is given by Lemma~\ref{lem:support}, with $t$ as the dominating point with respect to $s\text{-corr}(B_0,B_1)$. %The proof follows along the two proofs made for Zone IV of the non-intersecting case, with the added observation that $h_{\mathbb{AB}}(\alpha)=s$ for $\alpha \in[\gamma_{0},\gamma_{1}]\cup [\gamma_2, \gamma_3]$. 
Excluding the clockwise optimal region described above is essential here, as it forces $A_1$ to be outside of the wedge formed by the upper tangents from $A_0$ to $s\text{-circ}(B_0)$ and $s\text{-circ}(B_1)$ when $A_0$ is below both of the $B$ circles. This ensures that the path taken by $\mathbb A$ is convex.
\end{proof}

\subsection{Case 3}
\label{sec:the-rest-2b}

As Case 3 is highly constrained, most of the motions for this case are particularly simple. Figures \ref{case3a}, \ref{case3b-intersect-upper}, and~\ref{case3b-intersect-lower} exhibit possible configurations of Case 3. As before, we begin by defining the zones non-constructively, and then move on to more constructive descriptions.

Let $p_0$ and $p_1$ be the upper tangent points from $A_0$ to $s\text{-circ}(B_0)$ and $s\text{-circ}(B_1)$ respectively. The zones are defined by the following properties:

\begin{itemize}
\itemsep0em
\item [Zone I:] The set of points $q \in s\text{-corr}(B_0, B_1)$ that dominate $A_0$.
\item [Zone II:] The set of points $q \in s\text{-corr}(B_0, B_1)$ $A_0$ dominates.
\item [Zone III:] The set of points $q \in s\text{-corr}(B_0, B_1)$ where the tangent from $q$ to $s\text{-circ}(B_0)$ intersects $\overline{A_0p_1}$
\item [Zone IV:] The set of points $q \in s\text{-corr}(B_0, B_1)$ where the tangent from $q$ to $s\text{-circ}(B_1)$ intersects $\overline{A_0p_0}$.
\end{itemize}

We do not handle situations which reduce to Case 2. For example, if $A_1\in s\text{-corr}(B_0,B_1)$, is left of the tangent through $\overline{A_0p}$, and is above $s\text{-circ}(B_0)$, then we would be in Case 2. Similarly, if $A_1\in s\text{-corr}(B_0,B_1)$, is right of $\overline{A_0p_1}$, and above $s\text{-circ}(B_1)$, then we would also be in Case 2. 

Although Zone IV above is handled in Case 2, we keep it for symmetry.  Zones I-IV of Figures~\ref{case3a} and~\ref{case3b-intersect-upper} are defined
by the following curves: 

\begin{enumerate}
\item The two upper tangents from $A_0$ to $\text{circ}_s(B_0)$ and $\text{circ}_s(B_1)$ (through tangent points $p_i$). These tangents separate zone I from the rest of the zones. The tangent from $A_0$ to $p_1$ forms the left boundary of zone II if $A_0$ is below the tangent from the bottom of $\text{circ}_s(B_0)$ to the top of $\text{circ}_s(B_1)$. The tangent from $A_0$ to $p_0$ forms part of the right boundary of zone II.
\item The two horizontal tangents from $\text{circ}_s(B_0)$.
\item The lower tangent from $A_0$ to $\text{circ}_s(B_0)$ and $\text{circ}_s(B_1)$ (through tangent points $q_i$). The tangent from $A_0$ to $q_1$ (resp. $q_0$) form part of the right (resp. left) boundary for zone III (resp. zone IV). The tangent from $A_0$ to $q_0$ (resp $q_1$) forms the left (resp. right) boundary of zone II if $A_0$ is above the tangent from below $\text{circ}_s(B_0)$ to above $\text{circ}_s(B_1)$ (resp. above $\text{circ}_s(B_0)$ to below $\text{circ}_s(B_1)$).
\item The arc of $\text{circ}_s(B_0)$ (resp. $\text{circ}_s(B_1)$) from $p_0$ to $t_0$ (resp. $p_1$ to $t_1$). If the tangent from $A_0$ to $p_0$ (resp. to $p_1$) does not intersect $\text{circ}_s(B_1)$ (resp. $\text{circ}_s(B_0)$), then $t_0$ is $q_0$ (resp. $t_1$ is $q_1$). Otherwise, $t_0$ (resp. $t_1$) is the intersection point.
\item The arc of $\text{circ}_s(B_0)$ (resp. $\text{circ}_s(B_1)$) from $t_0$ to $q_0$ (resp. $t_1$ to $q_1$). These arcs forms part of the left and right boundaries of zone II.
\end{enumerate}

We now specify, for each zone, the location of $A_{\text{int}}$, and define $T_{0}$ and $T_{1}$ to be the lower tangent points of $B_0$
and $B_1$ to $\text{circ}_s(A_{\text{int}})$ respectively.

\begin{itemize}
\itemsep0em
\item [Zone I:] $A_{\text{int}}$ is the point $A_1$.

\item [Zone II:] $A_{\text{int}}$ is the point $A_0$.

\item [Zone III:] $A_{\text{int}}$ is the intersection point of the tangent from $A_1$ to the
$\text{circ}_s(B_0)$ and the tangent from $A_0$ to
$\text{circ}_s(B_1)$. 

\item [Zone IV:] $A_{\text{int}}$ is the intersection point of the tangent from $A_0$ to the
$\text{circ}_s(B_0)$ and the tangent from $A_1$ to
$\text{circ}_s(B_1)$. 
\end{itemize}

Our generic three-stage motion then becomes:
\begin{enumerate}
\item Move $\mathbb A$ on a straight line from $A_0$ to $A_{\text{int}}$
\item Move $\mathbb B$ from $B_0$ to $B_1$ avoiding $\text{circ}_s(A_{\text{int}})$. This involves moving $\mathbb B$ to $T_{0}$, rotating
$\mathbb B$ counter-clockwise about $A_0$ to $T_{1}$ in a range of angles
$[\beta_{0},\beta_{1}]$, and then moving $\mathbb B$ from $T_{1}$ to
$B_1$.
\item Move $\mathbb A$ on a straight line motion from $A_{\text{int}}$ to $A_1$.
\end{enumerate}
Note that in zone IV of Figure \ref{case3a}, all optimal counter-clockwise motions are of exactly the same from as zone III of Case 2. 

%Additionally, parts of Case 3 is already handled by Lemma~\ref{lem:optimal-clockwise}: If $s\text{-circ}(B_0)$ and $s\text{-circ}(B_1)$ do not intersect and $A_0$ is in region II of Figure \label{case3-part-repeat}, then we know the optimal motion must be net-clockwise. Furthermore, $A_0$ is in region I of the clockwise regions, so region II will automatically be handled if we show net-counter-clockwise optimal paths for region I. Hence we assume in the analysis below that $A_0$ is in region I whenever $s\text{-circ}(B_0)$ and $s\text{-circ}(B_1)$ do not intersect.

%%%%%%%%%%%%%%%% %%%%%%%%%%%%

\begin{figure}[ht]

\begin{center}
\includegraphics[width=\textwidth]{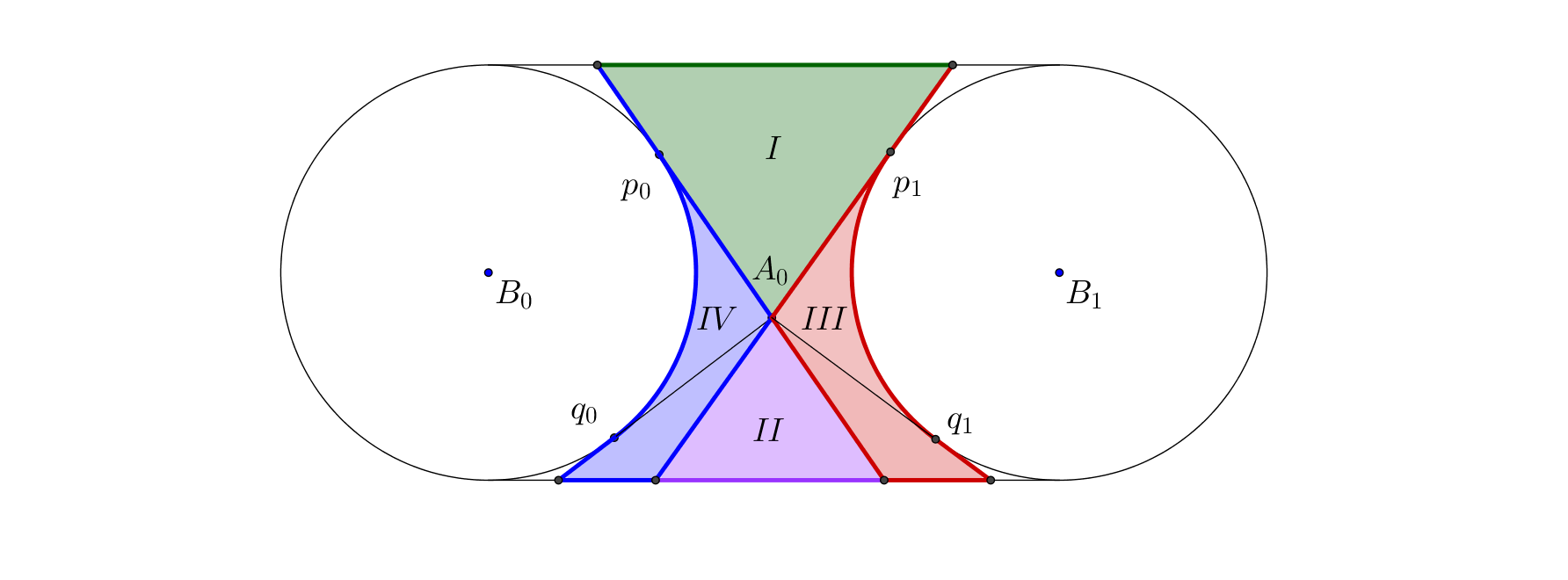}
\par\end{center}

\caption{Case 3, when $s\text{-circ}(B_0)$ and $s\text{-circ}(B_1)$ do not intersect. \label{case3a}}
\end{figure}

\begin{figure}[ht]

%\begin{subfigure}[b]{\textwidth}
  \begin{center}
  \includegraphics[width=0.9\textwidth]{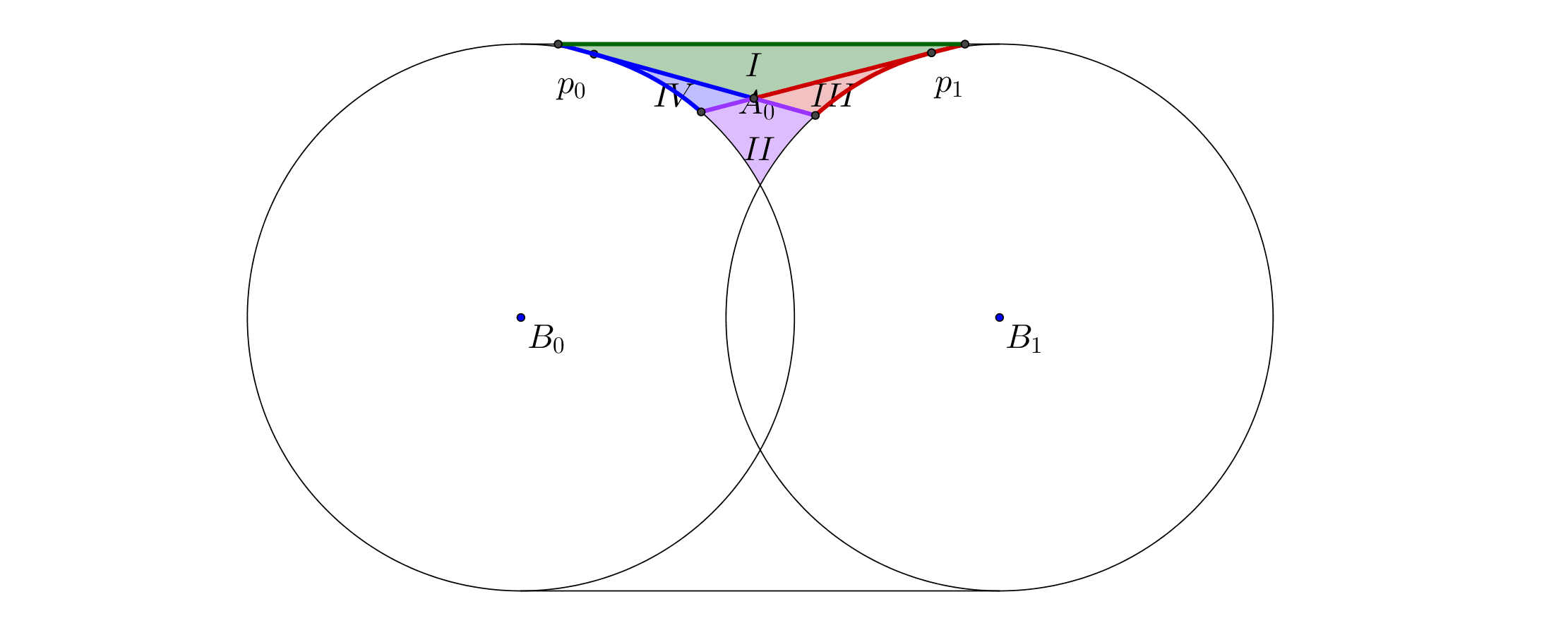}
  \end{center}
  \caption{Case 3, when $s\text{-circ}(B_0)$ and $s\text{-circ}(B_1)$ intersect and both $A_0$ and $A_1$ are above the $s\text{-circ}(B_i)$'s.}
  \label{case3b-intersect-upper}
%\end{subfigure}

\end{figure}

\begin{figure}
%\ContinuedFloat

%\begin{subfigure}[b]{\textwidth}
  \begin{center}
  \includegraphics[width=0.9\textwidth]{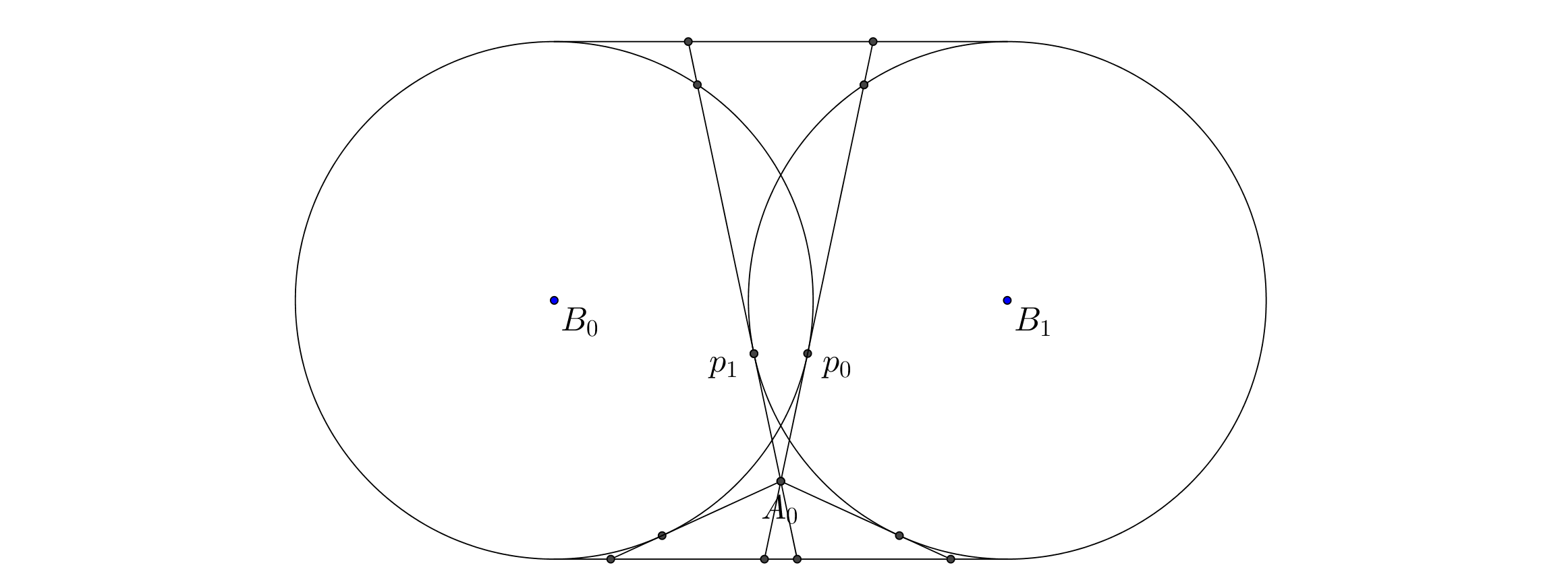}
  \end{center}
  \caption{Case 3, when $s\text{-circ}(B_0)$ and $s\text{-circ}(B_1)$ intersect and both $A_0$ and $A_1$ are below the $s\text{-circ}(B_i)$'s.}
  \label{case3b-intersect-lower}
%\end{subfigure}

%\caption{Case 3, when $s\text{-circ}(B_0)$ and $s\text{-circ}(B_1)$ intersect.\label{case3b}}
\end{figure}

\paragraph{Case 3, subcase 1: $s\text{-circ}(B_0)$ and $s\text{-circ}(B_1)$ do not intersect.}

\begin{proof}
In all cases (see Figure~\ref{case3a}), applications of Lemma~\ref{lem:support} will suffice. The proof of Zones III and IV are exactly the same as the proof for Case 2, Zone III. For Zones I and II, note that for all cases that Case 2 do not cover, $A_0$ must be reachable from $A_1$ by a straight-line. Hence there are no special cases and a single application of Lemma~\ref{lem:support} with either $A_0$ as the pivot (for Zone II) or $A_1$ as the pivot (for Zone I) suffices.
\end{proof}

\subsubsection{Subcase 2: $s\text{-circ}(B_0)$ and $s\text{-circ}(B_1)$ intersects}

When $s\text{-circ}(B_0)$ and $s\text{-circ}(B_1)$ intersect, observe that the constraints force either $A_0$ and $A_1$ to be both above the $B$ circles, or both below. This is because if $A_0$ was below the $s\text{-circ}(B_i)$'s and $A_1$ above, then we must be in Case 2 (after possibly swapping the initial and final positions).

\paragraph{$A_0$ and $A_1$ both above}
When $A_0$ and $A_1$ are both above the $B$ circles, we get Figure \ref{case3b-intersect-upper}. In this case, the same zones and proofs as the non-intersecting case apply.

\paragraph{$A_0$ and $A_1$ both below}
When $A_0$ and $A_1$ are both below the $B$ circles, we get Figure \ref{case3b-intersect-lower}. In this case, Lemma~\ref{lem:optimal-clockwise} shows that the motion must be net-clockwise. The clockwise zones have $A_0$ and $A_1$ in the ``both above'' case, which is handled above.

\section{Angle monotone motions}
\begin{figure}[h]
\begin{center}
  \includegraphics[width=0.9\textwidth]{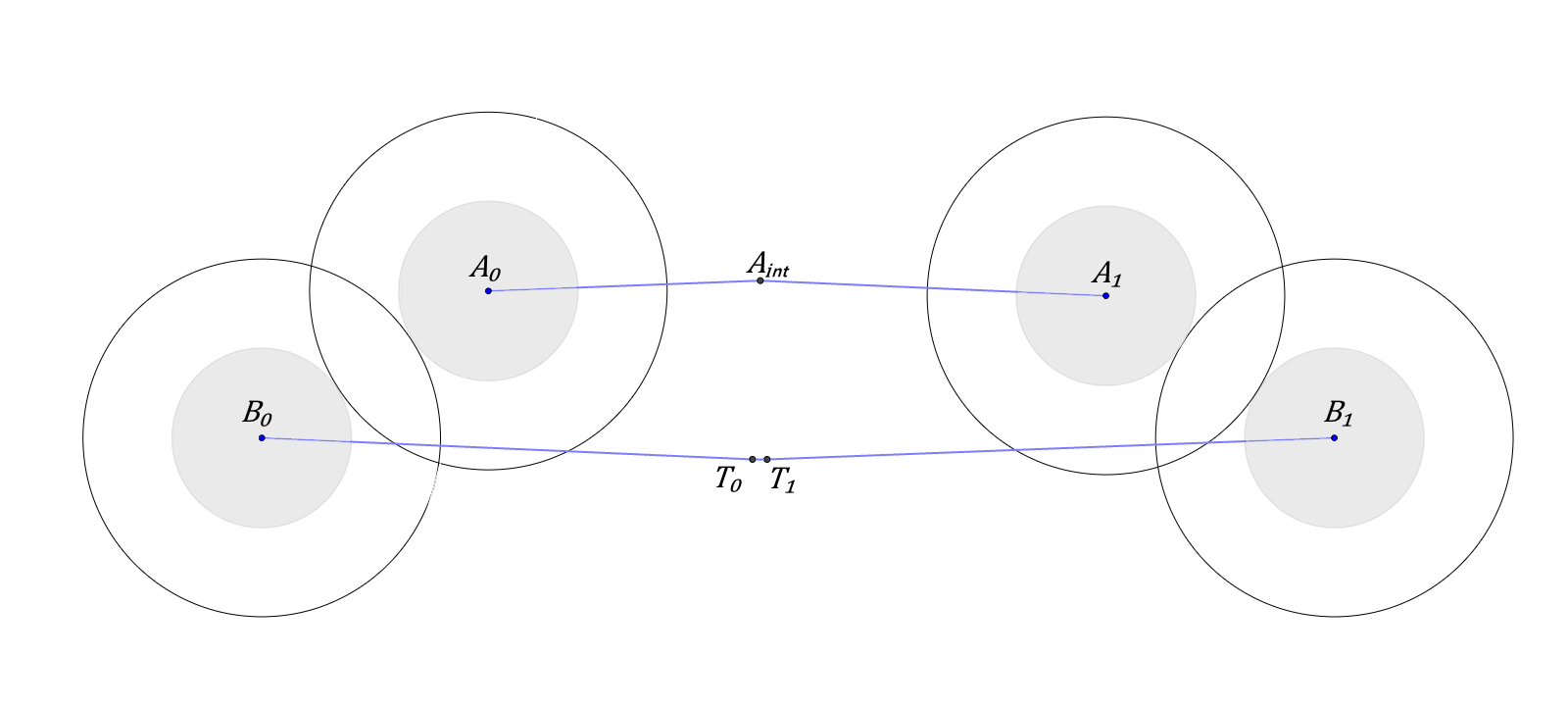}
  \end{center}
  \caption{\label{fig:non-monotone}}
\end{figure}

Up until now, we've stated all of our motions as \textbf{decoupled} motions where only one of $\mathbb A$ or $\mathbb B$ is moving at a time. However, we can produce angle monotone motions by simply \textbf{coupling} the optimal motions $m$ given in the previous sections. 

%[[ In what follows you seem to describe a simple transformation (betweem parallel configurations) that is clearly optimal. Doesn't it follow that if an optimal motion is not angle monotone then the trace between any two parallel configuations must be linear? If so, why do we need a lengthy proof, as asserted below?]]

To be precise, let $m$ be a motion such that $m(t_i)$ and $m(t_j)$ has the same angle. Then by coupling the motion $m$, we mean that we replace the submotion $m([t_i, t_j])$ with a straight-line path between $m(t_i)$ and $m(t_j)$. This process produces coupled angle monotone motions from decoupled ones. Most of the motions described in the previous section Section are angle monotone. The only situation in which non-angle monotonicity occurs in our decoupled motions is when $A_1$ is in Zone III of Case 3 above (see Figure~\ref{fig:non-monotone}). In all other cases, we have angle monotonicity for the decoupled motion as well, although the discs are possibly not be in contact for a single connected interval of time. 

One can also couple the motions to achieve both angle monotonicity and the property that the two discs are in contact for a single connected interval. This is obtained by following the trace of optimal motions outlined in the previous sections while keeping $\mathbb A$ and $\mathbb B$ as close together as possible. The proof, although not difficult, is lengthy as it requires examining the motions of each case in the previous Section. 

\section{Conclusions}
\label{sec:con}

Using the Cauchy surface area formula, we have presented and proved shortest collision-avoiding paths for two disc robots in a planar obstacle free environment. The path lengths are neatly characterized by a simple integral, and had the property that they could be \textbf{decoupled} so that only one disc is moving at any given time, or \textbf{coupled} so that the angle formed by a ray joining the two discs changes monotonically throughout the motion. The coupled motion has the additional property that discs are in contact for a connected interval of time, that is, once the discs move out of contact, they are never in contact again.

As far as we know, our tools are limited to the case when the robots are discs in 2D. Indeed, when the robots are spheres in 3D, even if the initial and final positions of the robot reside in a common plane, we have not been able to show that the shortest path stays within this plane (except in special cases). The 3D extension of the problem as well as the 2D problem with obstacles remain subjects for future exploration.

\bibliographystyle{plain}
\bibliography{motionplanning}

\end{document}